\documentclass[conference,letterpaper]{IEEEtran}

\usepackage[utf8]{inputenc} 
\usepackage[T1]{fontenc}
\usepackage{url}
\usepackage{ifthen}
\usepackage{cite}
\usepackage{hyperref}
\usepackage[cmex10]{amsmath} % Use the [cmex10] option to ensure complicance
\usepackage{bm}
\usepackage{amssymb, amsthm, amsfonts, mathtools, marvosym}
\newtheorem{definition}{Definition}
\newtheorem{theorem}{Theorem}
\newtheorem{lemma}{Lemma}
\newtheorem{corollary}{Corollary}

\newtheorem{assumption}{Assumption}
\newtheorem{remark}{Remark}

\theoremstyle{plain}
\newtheorem{example}{Example}
\usepackage{color}
\usepackage{xcolor}
\usepackage{graphicx}
\usepackage{subcaption}
\usepackage{extarrows}

\newcommand{\part}[2]{\frac{\partial #1}{\partial #2}}

\newcommand{\BN}{\mathbb{N}}
\newcommand{\R}{\mathbb{R}}
\newcommand{\E}{\mathbb{E}}
\newcommand{\pr}{\mathbb{P}}

\DeclareMathOperator{\grad}{grad}
\DeclareMathOperator{\Hess}{Hess}
\newcommand{\norm}[1]{\left\Vert{#1}\right\Vert}
\newcommand{\HS}[1]{\norm{#1}_{\mathrm{HS}}}

\begin{document}
\title{Convexity of Mutual Information\\
along the Fokker-Planck Flow} 

%%% Single author, or several authors with same affiliation:
\author{%
 \IEEEauthorblockN{Jiayang Zou, Luyao Fan, Jiayang Gao, Jia Wang}
 \IEEEauthorblockA{Department of Electronic Engineering \\
   Shanghai Jiao Tong University\\
   Shanghai, China\\
                     Email: qiudao@sjtu.edu.cn, fanluyao@sjtu.edu.cn, gjy0515@sjtu.edu.cn, jiawang@sjtu.edu.cn}
                    }

%%% Several authors with up to three affiliations:
% \author{%
%   \IEEEauthorblockN{Jiayang Zou}
%   \IEEEauthorblockA{Department of Electronic Engineering \\
%   Shanghai Jiao Tong University\\
%   Shanghai, China\\
%                     Email: qiudao@sjtu.edu.cn}
%   \and
%   \IEEEauthorblockN{Luyao Fan and Jia Wang}
%   \IEEEauthorblockA{Department of Electronic Engineering \\
%   Shanghai Jiao Tong University\\
%   Shanghai, China\\
%                     Email: \{fanluyao, jiawang\}@sjtu.edu.cn}
% }

\maketitle

%%%%%%
%% Abstract: 
%% If your paper is eligible for the student paper award, please add
%% the comment "THIS PAPER IS ELIGIBLE FOR THE STUDENT PAPER
%% AWARD." as a first line in the abstract. 
%% For the final version of the accepted paper, please do not forget
%% to remove this comment!
%%
%! 分成好多个appendix, 比如Proof of Lemma 4, 要写我们要用到Lemma10-12.
%! Arxiv版本每个证明后面加索引
\begin{abstract}
    We study the convexity of mutual information as a function of time along the Fokker-Planck flow. The results are generalizations of that along heat flow and Ornstein-Ulenbeck flow, which were established by A. Wibisono and V. Jog. We prove the existence and uniqueness of the classical solutions to a class of Fokker-Planck equations and then we obtain the second derivative of mutual information along the Fokker-Planck equation. 
    If the initial distribution is sufficiently
   strongly log-concave compared to the steady state, then mutual information always preserves convexity under suitable conditions. In particular, if there exists some time point at which the distribution is sufficiently strongly log-concave, then mutual information will preserve convexity after that time.
  \end{abstract}
  
  \section{Introduction}
  
  The Fokker-Planck equation (FPE) is an important partial differential equation to describe the evolution of probability density of stochastic processes, especially in statistical physics, machine learning and other fields. It can be understood and solved from different angles and methods.
  
  The Fokker-Planck equation can be regarded as the evolution equation of probability density of stochastic processes related to the It\^o stochastic differential equation  (SDE)\cite{jordan1998variational}. This means that for Markov processes satisfying the It\^o stochastic differential equation, the evolution of conditional probability density with time in a given initial state can be described by the Fokker-Planck equation\cite{sharma2010fokker}. Therefore, it is apt to designate this phenomenon as the Fokker-Planck flow (FP flow), highlighting the dynamic progression and transformation of the probability density in accordance with the Fokker-Planck equation. In addition, the Fokker-Planck equation can also be used to analyze the nonlinear system 
  for a variety of problems like demodulating the phase-locked frequency \cite{mayfield1973sequence}. 
  
  In the field of statistical physics and machine learning, the importance of the  Fokker-Planck equation lies in its ability to describe the continuity equation of density evolution, in which the change of density is completely determined by a time-varying velocity field, which in turn depends on the current density function\cite{shen2022self}. This self-consistency makes Fokker-Planck equations the basis of designing latent functions and neural network parametric models, and then generates the whole density trajectory to approximate the solution of Fokker-Planck equations \cite{shen2022self}. In specific application fields, such as plasma physics and nonlinear filtering, the Fokker-Planck equation is also used to describe specific physical phenomena, for example, the interaction between RF waves and plasma, and the probability density of states given by observation results\cite{yau2004computation}. These applications show the flexibility and effectiveness of Fokker-Planck equations in solving practical physical problems.
  
  In recent years, A. Wibisono and V. Jog have proposed the utilization of the FPE for channel modeling, and extended the channel originally modeled by discrete Markov chains to be modeled by continuous Markov chains \cite{wibisono2017information}. This paradigm shift conceptualizes the channel as a dynamically evolving system over time, so that the channel can be analyzed by SDE methods. Particularly, the evolution of the probability density function along the trajectory defined by an SDE can be characterized by the FPE. Consequently, channels that adhere to this description are referred to as Fokker-Planck channels. Furthermore, Wibisono and Jog have extended the I-MMSE (Information-Minimum Mean Square Error) identity\cite{guo2005mutual,wu2011derivative,guo2011estimation}, or equivalently, De Bruijn's identity\cite{stam1959some,emamirad2023bruijn}, originally formulated for the scenario of an additive white Gaussian noise channel, which delineates the progression of Brownian motion and can be represented by a straightforward SDE, to the more general Fokker-Planck channel.
  
  In the conventional information theory, it is well-known that the mutual information gradually decreases by data processing inequality in the process of transmission on Markov chains\cite{cover2006elements}. And it is a convex function for a fixed initial probability density function. While in the perspective of Fokker-Planck channels, the mutual information is a function over time. We have known that the mutual information is decreasing along the Fokker-Planck process, which means the first time derivative is negative\cite{wibisono2017information}. And it has been proved that if the initial distribution is log-concave, then mutual information is always a convex function of time along the heat flow, which is a special case of the Fokker-Planck flow\cite{wibisono2018convexity}. 
  Moreover, if the initial distribution is sufficiently
  strongly log-concave compared to the target Gaussian measure, then mutual information is always a convex function of time along the Ornstein-Uhlenbeck (OU) flow\cite{wibisono2018convexityOU}. However, there is a lack of relevant research on the properties of mutual information along the FP flow, even the simplest linear FPE, since it has no explicit solution like the heat flow or the OU flow. Hence we need to explore new methods that circumvent the dependence on explicit solutions to study the convexity of mutual information.
  
  In this paper we study the convexity of mutual information along the FP flow to generalize the conclusions in cases of the heat flow, with some suitable assumptions. It is known that the properties of the FP flow are closely linked to entropy. By optimal transport theory, the Fokker-Planck flow is the gradient flow in the Wasserstein space (the space of probability distributions with the Wasserstein metric structure) associated with the energy functional taken the form by relative entropy with respect to the steady distribution\cite{villani2021topics,villani2009optimal,figalli2021invitation,santambrogio2015optimal,ambrosio2005gradient}. 
  % Note that the Fokker-Planck flow in optimal transport theory is the simplest case of FPEs, which is a linear partial differential equation, and it can not take the general form otherwise it is not a gradient flow of entropy. For the sake of utilizing the derivative formula of energy functional along gradient flow, we only pay attention to the simple case of linear FPEs.
  
  Our first main result is that we prove the uniqueness and existence of the classical solution to the FP flow and obtain the second derivative of mutual information along the FP flow, contributing to studying the convexity of mutual information over time. Then by the theory of spatial log-concavity of solutions to parabolic systems\cite{ishige2010new,ishige2020spatial,ishige2020logconcavity}, we prove that if the initial distribution is sufficiently strongly log-concave compared to the steady state, and the initial distribution $u_0$ is nonnegative and continuous, and the potential function $V(x)$ satisfies
  $\Delta V \leqslant \frac{1}{2}\Vert \nabla V\Vert^2$, 
  then mutual information always preserves convexity. In particular, if there exists some time point at which the distribution is sufficiently strongly log-concave compared to the steady state, then mutual information preserves convexity after this time under the same conditions. 
  
  The remainder of the paper is organized as follows: Section \ref{Existence and Uniqueness of Classical Solutions to FP Flow} reviews the fundamentals of the FPE and proves the existence and uniqueness of classical solutions to the FP flow with suitable assumptions, as the foundation of our work. Section \ref{Background of Log-concavity, Entropy and Mutual Versions} recapitulates some useful properties of log-concave functions and the log-convexity of solutions to parabolic systems. It also reviews the definitions of entropy, Fisher information, and their mutual versions. Section \ref{Convexity of mutual information} presents the main results and proofs of the convexity of mutual information along the FP flow. We conclude with Section \ref{sec:conclusion}, which includes a discussion of open problems and future directions. 
  
  Detailed proofs of results in Section \ref{Existence and Uniqueness of Classical Solutions to FP Flow} and \ref{Convexity of mutual information} are provided in Appendix.
  \section{Existence and Uniqueness of Classical Solutions to the FP Flow}
  \label{Existence and Uniqueness of Classical Solutions to FP Flow}
  %\subsection{SDEs and Fokker-Planck flow}\label{subsectionA}
  
  Consider the general Fokker-Planck channel that outputs a real-valued stochastic process $(X_t)_{t\geqslant 0}$ in 
  $\mathbb{R}^N$ with respect to the SDE
  \begin{equation}
      dX_t = a(X_t,t)dt+\sigma(X_t,t)dW_t
  \end{equation}
  where $(W_t)_{t\geqslant 0}$ is the standard Brownian motion in $\R^N$, $a(x,t)$ is the drift coefficient and $\sigma(x,t)\succ0$ is the diffusion coefficient, where $\sigma(x,t)\in \R^{N\times N}$. Note that we define $A\succ B$ if and only if $A-B$ is positive definite, and define $A\succeq B$ if and only if $A-B$ is positive semi-definite. 
  
  We denote the density function of $X_t$ for a fixed time $t\geqslant 0$ by $\mu(x,t)$ or $\mu_t(x)$ over space $x\in\R^N$, then the density function satisfies the Fokker-Planck equation
  \begin{equation}
      \begin{aligned}
          \partial_t\mu=\dfrac{\partial \mu}{\partial t}=\Delta (D\mu)-\nabla\cdot(a\mu ).
      \end{aligned}
  \end{equation}
  Here $D(x,t)=\sigma(x,t)^\top\sigma(x,t)/2$, $\nabla\cdot = \sum_{i=1}^{n} \part{}{x_i}$ is the divergence and $\Delta= \sum_{i=1}^{n}\frac{\partial^2}{\partial x_i^2}$ is the Laplacian operator.
  
  The choice $a\equiv 0$ and $\sigma \equiv \sqrt{2}$ generates the heat flow, i.e. the Gaussian channel. Also, if we choose $a(x,t)=-\alpha x, \alpha>0$ and $\sigma \equiv 1$, this equation generates the OU flow. In this paper, we consider the case where $a(x,t)$ is independent of $t$ and $\sigma\equiv\sqrt{2}$. Define the potential function $V(x)$ whose negative gradient is $a(x)$, that is,
  $\nabla V = -a(x)$
  with the potential $V=V(x)$ satisfying $e^{-V}\in L^1(\mathbb{R}^N)$. Then we have the Fokker-Planck flow
  \begin{equation}\label{FPflow}
      \begin{aligned}
  \frac{\partial \mu}{\partial t}=\Delta \mu +\nabla \cdot \left( \mu \nabla V \right). 
      \end{aligned}
  \end{equation}
  This result requires that $a(x,t)$ and $\sigma(x,t)$
  satisfy appropriate regularity and growth
  conditions, such as smoothness and Lipschitz properties\cite{mackey2011time}. 
  
  If we take $\gamma_t = e^{\frac{V}{2}}\mu_t$, then the Fokker-Planck flow \eqref{FPflow} can be simplified by 
  \begin{equation}\label{schrodinger2}
  \frac{\partial \gamma}{\partial t}=\left( \frac{1}{2}\Delta V-\frac{1}{4}\left\| \nabla V \right\| ^2 \right) \gamma +\Delta \gamma 
  \end{equation}
  which is an (imaginary-time) Schr\"odinger equation. Let $c(x)=\frac{1}{4}\norm{\nabla V}^2-\frac{1}{2}\Delta V$. Throughout this paper we assume the following conditions:
  \begin{assumption}\label{Assumption for potential}
      $V(x)$ is nonnegative and smooth, and $c(x)$ is convex with a lower bound.
  \end{assumption}
  This assumption can hold if we consider the Ornstein-Uhlenbeck process.
  \begin{assumption}\label{Assumption for initial}
      $\gamma_0(x)$ is bounded, that is, $\mu_0(x)\lesssim e^{-\frac{V(x)}{2}}$. Note that $f\lesssim g$ means there exists a constant $C$ such that $f\leqslant Cg$.
  \end{assumption}
  
  \begin{remark}
  By considering $\Tilde{\gamma_t}:=\gamma_t e^{-\beta t}, \beta>\inf_{x\in \R^N} c(x)$, we have
  \begin{equation}
      \begin{aligned}
          \label{schrodinger3}
  \frac{\partial \Tilde{\gamma}}{\partial t}+\left( c(x)+\beta \right) \Tilde{\gamma} =\Delta \Tilde{\gamma}.
      \end{aligned}
  \end{equation}
  Since the existence and uniqueness of equations \eqref{schrodinger2} and \eqref{schrodinger3} are equivalent, and $\Tilde{\gamma}$ has the same log-concavity property as $\gamma$, we only need to consider the case where $c(x)$ is nonnegative.
  \end{remark}
  By large-time behaviors of Fokker-Planck equations\cite{bogachev2022fokker,risken1996fokker,arnold2008large}, we know that the Fokker-Planck flow \eqref{FPflow} has the unique normalized steady state $\mu_{\infty}=e^{-V}$, and the relative entropy between $\mu_t$ and $\mu_{\infty}$ converges to $0$ exponentially as $t\to \infty$\cite{jungel2016entropy}. Moreover, the $L_1$ norm of $\mu_t-\mu_{\infty}$ and the Wasserstein distance between them also converges to $0$ exponentially as $t\to \infty$\cite{villani2009optimal}.
  
  In this paper, we can assume that solutions to this Fokker-Planck flow and the Schr\"odinger equation \eqref{schrodinger2} exist and sufficiently smooth.
  Moreover, to facilitate the technique of integration by parts, it is further assumed that the solution is a rapidly decreasing function.
  While in some cases, the existence and uniqueness of classical solutions can be substantiated through the application of stochastic analysis methods, which is shown in Appendix \ref{appendix: Existence and Uniqueness to the Cauchy Problem}.
  \begin{theorem}\label{Thm:Existence and Uniqueness to the Cauchy Problem}
      If Assumption \ref{Assumption for potential} and \ref{Assumption for initial} hold, the Fokker-Planck flow \eqref{FPflow} in the entire space has the unique classical solution
       $$\mu(x,t)=
   e^{-\frac{V\left( x \right)}{2}}\E_{x/\sqrt{2}}\left[
          \gamma_0(\sqrt{2}W_{t})e^{-\int_{0}^{t}c(\sqrt{2}W_r)dr}
          \right]
  $$
  where $
   c\left( x \right) =\frac{1}{4}\left\| \nabla V \right\| ^2-\frac{1}{2}\Delta V
  , \gamma_0=e^{\frac{V}{2}}\mu_0$. Furthermore, $\mu(x,t)\in C^{1,\infty}((0,\infty)\times \R^N).$
  
  Then this solution can be expressed as the convolution of a kernel function with $\mu_0$
  \begin{equation*}
      \begin{aligned}
          \mu(x,t)&=e^{-\frac{V(x)}{2}}\int_{\R^N} K(x,y,t)\gamma_0(y) dy\\
          &=\int_{\R^N} K(x,y,t)e^{\frac{V(y)-V(x)}{2}}\mu_0(y) dy.
      \end{aligned}
  \end{equation*}
  
  \end{theorem}
  \begin{proof}
  See Appendix \ref{appendix: Existence and Uniqueness to the Cauchy Problem}.
    \end{proof}
  Alternatively, we can also represent it by an operator semigroup\cite{li2011perelman}
  $$\mu(x,t)=Q_t \mu_0 := e^{-V}P_t(e^{V}\mu_0)$$
  where $P_t=e^{tL}$ is the heat semigroup generated by the Witten Laplacian $L=\Delta-\nabla V\cdot \nabla$.

  Indeed, if this FP flow only exists in a bounded domain, the solution can be also expressed as this way according to subsection 2.4.4 of \cite{avramidi2015heat}
  \begin{equation}
      \begin{aligned}
      \label{kernel function in a bounded domain}
          \mu(x,t)=\int_{\Omega} K_{\Omega}(x,y,t)e^{\frac{V(y)-V(x)}{2}}\mu_0(y) dy
      \end{aligned}
  \end{equation}
  where $K_{\Omega}$ is the kernel function corresponding to this equation in a bounded domain $\Omega$ with smooth boundary.
  
  \section{Log-convavity and Entropy Related Quantities}
  \label{Background of Log-concavity, Entropy and Mutual Versions}
  \subsection{Log-concavity of Solutions to Parabolic Systems}
  \label{Preservation of log-concavity}
  \begin{definition}
  A nonnegative function $u$ in $\R^N$ is said log-concave if 
  \begin{equation}
      u((1-\mu)x+\mu y)\geqslant u(x)^{1-\mu}u(y)^{\mu}
  \end{equation}
  for $\mu\in [0,1]$ and $x,y\in\R^N$ such that $u(x)u(y)>0$\cite{saumard2014log}
  \end{definition}
  This is equivalent to the following:
  \begin{enumerate}
      \item The set $S_{u}:=\{x\in\R^N : u(x)>0\}$ is convex and $\log u$ is concave in $S_u$.
      \item $u=e^{-\phi}$ where $\phi$ is a convex function.
  \end{enumerate}
  
  Log-concavity is a very useful variation of concavity and plays an important role in various fields such as PDEs, geometry, probability, statics and so on. Here we give some significant properties of log-concavity used in this paper:
  \begin{enumerate}
      \item The normal distribution and multivariate normal distributions are log-concave.
      \item Convolution preserves log-concavity. That is,  if $f$ and $g$ are log-concave, then 
      $$(f\ast g)(x)=\int f(x-y)g(y) dy$$
      is log-concave. 
      \item Log-concavity is preserved under convergence in distribution\cite{dharmadhikari1988unimodality}.
      
  \end{enumerate}

  By exploring the log-concavity of the Gauss kernel and the Prekópa-Leindler inequality, Brascamp and Lieb proved that log-concavity is preserved by the heat flow\cite{brascamp1976extensions}. 
  Furthermore, Lee and Vázquez proved that if the initial distribution is a bounded nonnegative function in $\R^N$ with compact support, then it is eventually log-concave along the heat flow\cite{lee2003geometrical}.
  
  Additionally, Ishige et al. introduced more generalized concepts of concavity, including spatial concavity\cite{ishige2020spatial} and F-concavity\cite{ishige2022characterization}, to examine the strongest or weakest concavity property preserved throughout the heat flow. They conducted a deep investigation into the log-convexity of solutions to parabolic systems.
  
  \begin{lemma}[]
  (Theorem 1.1 in \cite{ishige2020spatial})
  Let $\Omega$ be a bounded convex domain in $\R^N$ and $d_1,d_2>0$. Let $D:=\Omega\times (0,\infty)$,$(u,v)\in  C^{2,1}(D:\R^2)\cap C(\Bar{D}:\R^2) $ satisfy 
  \begin{equation*}
      \begin{aligned}
          \begin{cases}
              \partial_{t} u-d_1\Delta u +f(x,t,u,v,\nabla u) =0 & in\;\; D\\
              \partial_{t} v-d_2\Delta v +g(x,t,u,v,\nabla v) =0 & in \;\;D\\
              u,v\geqslant 0 & in \;\;D\\
              u(x,t)=v(x,t) =0 & on \;\; \partial\Omega\times [0,\infty)\\
              u(x,0)=u_0(x), v(x,0)=v_0(x) & in \;\;D
          \end{cases}
      \end{aligned}
  \end{equation*}
  where $f,g$ are nonnegative continuous functions in $D\times [0,\infty)^2\times \R^N
  $. Assume the following conditions:
  \begin{enumerate}
      \item The viscosity comparison principle holds for the above system.
      \item The functions 
      $$\mathfrak{f}_{t,\theta}(x,r,s):=e^{-r}f(x,t,e^r,e^s,e^r\theta)$$
      and
      $$\mathfrak{g}_{t,\theta}(x,r,s):=e^{-s}g(x,t,e^r,e^s,e^s\theta)$$
      are convex in $\Omega\times (0,+\infty)^2$ for every fixed $t>0$ and $\theta\in\R^N$.
  \end{enumerate}
  
  Then $\log u(\cdot,t)$ and $\log v(\cdot,t)$ are concave in $\Omega$ for every fixed $t>0$, provided that
  $\log u_0$ and $\log v_0$ are concave in $\Omega$.
  \end{lemma}
  
  Consider an imaginary-time Schr\"odinger equation
  \begin{equation}\label{schrodinger}
  \frac{\partial \gamma}{\partial t}+c(x) \gamma -\Delta \gamma =0
  \end{equation}
  where $c(x)$ is nonnegative and sufficiently smooth.
  Based on Assumption \ref{Assumption for potential} and \ref{Assumption for initial}, it is established that the equation \eqref{schrodinger} possesses a unique classical solution. Furthermore, since the maximum principle holds to this parabolic PDEs—specifically, the classical comparison principle—it is demonstrated that the viscosity comparison principle also holds for equation \eqref{schrodinger}\cite{crandall1992user}.

  \begin{corollary}\label{Schr\"odinger equation_log concave}
      Let $\Omega$ be a bounded convex domain in $\R^N$ and let $u\in C^2(D)\cap C(\Bar{D}), D:=\Omega\times (0,\infty)$ satisfy
      \begin{equation*}
          \begin{cases}
              \partial_t u + c(x)u-\Delta u = 0 & in \;\; D\\
              u(x,t) = 0 & on \;\; \partial\Omega\times[0,\infty)\\
              u(x,0)=u_0(x) & in \;\; \Omega
          \end{cases}
      \end{equation*}
  where $u_0$ is a nonnegative continuous function on $\Bar{\Omega}$, and Assumption \ref{Assumption for potential} holds. Then $u_t$ is log-concave in $\Omega$ for every fixed $t>0$, if $u_0$ is log-concave in $\Omega$.

  \end{corollary}
  
  \subsection{Fundamental Quantities and Mutual Versions}
  \label{Fundamental quantities and mutual versions}
  Let $X$ be a random variable in $\R^N$ and its probability density function $\mu$ is smooth and positive. 
  \begin{definition}\cite{wibisono2018convexity}
      The (differential) entropy of $X\sim \mu$ is 
      $$H(X) = -\int_{\R^N}\mu(x)\log \mu(x) dx. $$
       The Fisher information of $X\sim \mu$ is 
      $$J(X) = \int_{\R^N}\mu(x)\Vert \nabla\log\mu(x)\Vert^2 dx= \int_{\R^N}
      \dfrac{\Vert \nabla\mu(x)\Vert^2}{\mu(x)}
      dx.$$
      The second-order Fisher information of $X\sim \mu$ is
      $$K(X)=\int_{\R^N}\mu(x)\norm{\nabla^2 \log\mu(x)}^2_{\mathrm{HS}} dx.$$
      Here $\HS{A}^2 =\sum_{i,j=1}^{N} A_{ij}^2= \sum_{i=1}^{N}\lambda_i(A)^2$ is the Hilbert-Schmidt (or Frobenius) norm of a symmetric matrix $A=(A_{ij})\in\R^{N\times N}$ with eigenvalues $\lambda_i(A)\in\R$ and $\nabla^2 u(x) :=\left(\dfrac{\partial^2 u}{\partial x_i \partial x_j}\right) \in \R^{N\times N}$. Furthermore we can define the Hilbert-Schmidt inner product for two matrices with identical size
      $$\left<A,B\right>_{HS} := \sum_{i,j=1}^{N} A_{ij} B_{ij}.$$
  \end{definition}
  
  In optimal transport and gradient flow theory, we generally study the entropy and the Fisher information over a reference probability measure, denote by $\nu$ on $\R^N$.
  
  \begin{definition}\cite{wibisono2018convexityOU}
      The relative entropy of a probability measure $\mu$ with 
      respect to $\nu$ is 
      $$H_{\nu}(\mu) = \int_{\R^N}\mu\log \dfrac{\mu}{\nu} dx. $$
  This is also known as the Kullback-Leibler (KL) divergence.
      The relative Fisher information of $\mu$ with 
      respect to $\nu$ is 
      $$J_{\nu}(\mu) = \int_{\R^N}\mu 
      \norm{\nabla\log\dfrac{\mu}{\nu}}^2 dx.
      $$
      The relative second-order Fisher information of $\mu$ with 
      respect to $\nu$ is
      $$K_{\nu}(\mu)=\int_{\R^N}\mu \HS{\nabla^2 \log\dfrac{\mu}{\nu}}^2 dx.$$
  \end{definition}
  
  \begin{definition}\cite{wibisono2018convexityOU}
      Given a functional $F(Y)\equiv F(\mu_Y)$ of a random variable $Y\sim \mu_Y$, we can define its mutual version $F(X;Y)$ for a joint random variable $(X,Y)\sim \mu_{XY}$ by
      \begin{equation}
          F(X;Y) = F(Y|X)-F(Y)
      \end{equation}
  where $F(Y|X)= \int_{\R^N} \mu_{X}(x) F(\mu_{Y|X}(\cdot|x)) dx$ is the expectation of $F$ on the conditional random variables $Y|\{X=x\}\sim \mu_{Y|X}(\cdot |x)$, averaged over $X\sim \mu_{X}$. 
  \end{definition}
  
  Note that the mutual version exclusively captures the nonlinear component, so two functionals that differ solely by a linear function will possess identical mutual versions.
  \begin{definition}
  The mutual information is equal to the mutual version of negative entropy $-H(\mu)$ or the relative entropy $H_{\nu}(\mu)$, which is defined by
  $$I(X;Y)=H(Y)-H(Y|X)=H_{\nu}(X;Y).$$
  The mutual relative Fisher information is defined by
  $$J_{\nu}(X;Y)=J_{\nu}(Y|X)-J_{\nu}(Y)$$
  The mutual relative second-order Fisher information is defined by
  $$K_{\nu}(X;Y)=K_{\nu}(Y|X)-K_{\nu}(Y).$$
  \end{definition}
  Note that $J_{\nu}(X;Y)$ is equal to the backward Fisher information (or the statistical Fisher information) $\Phi(X|Y)$ that is defined by \cite{wibisono2017information}:
  $$\Phi(X|Y)=\int_{\R^N\times\R^N} \mu_{XY}\norm{\nabla_y \log \mu_{X|Y}}^2 dx dy.$$
  Hence $J_{\nu}(X;Y)$ is always nonnegative.
  % \begin{example}\cite{wibisono2018convexityOU}
  %     For any reference measure $\nu$, the relative entropy 
  %     $H_{\nu}(\mu)$
  %     differs from the negative entropy $-H(\mu)=\int_{\R^N} \mu\log\mu dx$ by the linear term(over $\mu$) $\int_{\R^N} \mu \log \nu dx$. Therefore, the mutual version of relative entropy is equal to the mutual version of negative entropy, which is the mutual information in information theory
  %     $$H_{\nu}(X;Y)=I(X;Y)=H(Y)-H(Y|X).$$
  % \end{example}
  % \begin{example}\cite{wibisono2018convexityOU}
  % The mutual relative Fisher information is defined by
  % $$J_{\nu}(X;Y)=J_{\nu}(Y|X)-J_{\nu}(Y)$$
  % which is always nonnegative. Indeed, it is independent of the reference measure $\nu$, since $J_{\nu}(X;Y)$ is equal to the backward Fisher information (or the statistical Fisher information) $\Phi(X|Y)$ that is defined by:
  % $$\Phi(X|Y)=\int_{\R^N\times\R^N} \mu_{XY}\norm{\nabla_y \log \mu_{X|Y}}^2 dx dy.$$
  % Here we need to note that the mutual relative Fisher information do not has the symmetric property, that is, $J_{\nu}(X;Y)$ is generally not equal to $J_{\nu}(Y;X)$. And the variable of reference measure $\nu$ is identical to $\mu_{Y}$. 
  % %\textcolor{red}{See Appendix A.}  
  % \end{example}
  
  % \begin{example}\cite{wibisono2018convexityOU}
  % The mutual relative second-order Fisher information is defined by
  % $$K_{\nu}(X;Y)=K_{\nu}(Y|X)-K_{\nu}(Y).$$
  % But it can be negative, unlike $J_{\nu}(X;Y)$ which is always positive.
  % \end{example}
  
  \section{Convexity of Mutual Information}\label{Convexity of mutual information}
  \subsection{Derivatives of Mutual Information along the FP Flow}
  
  In this subsection, we apply definitions of fundamental quantities to the joint random variable $(X,Y)=(X_0, X_t)$, where $X_0\sim \mu_{X_0}$ represents the initial distribution, and $X_t\sim \mu_{X_t}$ signifies the distribution at a given time $t\geqslant 0$ along the FP flow originating from $\mu_{X_0}$. For the purpose of notation simplification, $\mu_{X_0}$ and $\mu_{X_t}$ are respectively abbreviated as $\mu_0$ and $\mu_t$. Furthermore, the distribution of the joint random variable $(X_0,X_t)$ is denoted by $\mu_{0t}$ or $\mu_{t0}$. Additionally, the distribution of the conditional random variable $X_t|\{X_0=x_0\}$ is represented by $\mu_{t|0}(\cdot|x_0)$.
  \begin{lemma}\label{derivatives of entropy}
  Along the FP flow for the reference measure $\nu =\mu_{\infty}= e^{-V}$, we have
  \begin{equation}
      \begin{aligned}
  \frac{d}{dt}H_{\nu}\left( \mu _t \right)&=-J_{\nu}\left( \mu _t \right) =\int_{\mathbb{R} ^N}{\mu _t}\left\| \nabla \log \frac{\mu _t}{\nu} \right\| ^2dx\\
   \frac{d^2}{dt^2}H_{\nu}\left( \mu _t \right)&=2K_{\nu}\left( \mu _t \right) +2G_{\nu}(\mu_t)\\
  \frac{d^2}{dt^2}H_{\nu}\left( \mu _{t|0} \right)&=2K_{\nu}\left( \mu _{t|0}  \right) +2\int_{\mathbb{R}^{n}}\mu_0 G_{\nu}(\mu_{t|0}) dx_0
      \end{aligned}
  \end{equation}
  where 
  $$\displaystyle G_{\nu}(\mu)=
  \int_{\mathbb{R} ^N}{\mu}\left< \left( \nabla ^2V \right) \nabla \log \frac{\mu}{\nu}, \nabla \log \frac{\mu}{\nu} \right> dx
  $$
  and note that the variables of $\nu$ and $V$ is represented by $x_t$. Here $\left<\cdot,\cdot\right>$ denotes the inner product.
  \end{lemma}
  \begin{proof}
See Appendix \ref{appendix: derivatives of entropy}.
  \end{proof}

  This lemma is a direct corollary of Theorem 24.2 in \cite{villani2009optimal} if we select function $U(x)=x\log x$. Then the energy functional $U_{\nu}$ degenerates into the relative entropy $H_{\nu}$. Therefore we can calculate derivatives of mutual information along the FP flow.
  \begin{theorem}\label{derivatives of mutual information}
      Along the FP flow for the reference measure $\nu =\mu_{\infty}= e^{-V}$, we have
  %     \begin{equation*}
  %         \begin{aligned}
  %             \dfrac{d}{dt} I(X_0;X_t)&= -J_{\nu}(\mu_0;\mu_t)\leqslant 0\\
  %             \dfrac{d^2}{dt^2} I(X_0;X_t)&=
  %              2\Psi \left( \mu _0|\mu _t \right)
  %             +4\int_{\R^{2n}}{\mu _{0,t}}
  % \left< \nabla^2\log\gamma_t, \nabla ^2\log \mu _{0|t} \right> _{HS}dx_0dx_t\\
  % &=2\Psi \left( \mu _0|\mu _t \right)- 4\int_{\mathbb{R}^{2n}}{\mu _{0,t}}\left( \nabla \log \mu _{0|t} \right) ^{\top}\left( \nabla^2\log\gamma_t \right) \left( \nabla \log \mu _{0|t} \right) dx_0dx_t
  %         \end{aligned}
  %     \end{equation*}
  \begin{equation*}
          \begin{aligned}
              &\dfrac{d}{dt} I(X_0;X_t)= -J_{\nu}(\mu_0;\mu_t)\leqslant 0\\
              &\dfrac{d^2}{dt^2} I(X_0;X_t)=\\
               &\quad2\Psi \left( \mu _0|\mu _t \right)
              +4\int_{\R^{2N}}{\mu _{0,t}}
  \left< \nabla^2\log\gamma_t, \nabla ^2\log \mu _{0|t} \right> _{HS}dx_0dx_t\\
  &=2\Psi \left( \mu _0|\mu _t \right)-\\
  & \quad4\int_{\mathbb{R}^{2N}}{\mu _{0,t}}\left( \nabla \log \mu _{0|t} \right) ^{\top}\left( \nabla^2\log\gamma_t \right) \left( \nabla \log \mu _{0|t} \right) dx_0dx_t.
          \end{aligned}
      \end{equation*}
  Here $\gamma_t:=e^{\frac{V}{2}}\mu_t$ and
  the backward second-order Fisher information of $\mu_0$ given $\mu_t$ is defined by
  $$\Psi(\mu_0|\mu_t)=\int_{\R^{2N}}
  \mu_{t,0}\HS{\nabla^2\log\dfrac{\mu_{0|t}}{\nu}}^2  dx_0 dx_t.
  $$
  Sometimes we replace $\R^{2N}$ with $\R^N_0\times\R^N_t$ to distinguish the spaces of two variables $x_0, x_t$. 
  \end{theorem}
  \begin{proof}
See Appendix \ref{appendix: derivatives of mutual information}.
  \end{proof}
  
  \subsection{Convexity of Mutual Information along the FP Flow}
  We elucidate our main results concerning
  the convexity properties of mutual information along the FP flow. For the entirety of this subsection, we designate $X_t \sim \mu_t$ to represent the FP process evolving from an initial state $X_0 \sim \mu_0$. The idea of proofs follows the work\cite{ishige2020logconcavity} of Ishige et al.
  
  \begin{lemma}\label{step1}
      Let Assumption \ref{Assumption for potential},\ref{Assumption for initial} hold. Let $\Omega$ be a bounded smooth convex domain, $\mu_0\in C(\Bar{\Omega})$ and $\mu_0 =0 $ on $\partial \Omega$. Additionally we suppose $X_0\sim \mu_0$ is $\frac{V}{2}$-relatively log-concave where $V(x)$ is the potential function of FP flow \eqref{FPflow}. Then the mutual information $I(X_0;X_t)$ is convex over $t$ perpetually.
  \end{lemma} 
  
  \begin{proof}
See Appendix \ref{appendix: proof of step1}.
  \end{proof}
  We say $X\sim \mu$ is $\lambda$-relatively log-concave for some strictly nonnegative function $\lambda(x)$ if $-\nabla^2 \log\mu(x)\succeq \nabla^2 \lambda(x)$. Furthermore, we prove that the condition $\mu_0=0$ on $\partial \Omega$ is indeed redundant.
  
  \begin{lemma}\label{step2}
      Let Assumption \ref{Assumption for potential},\ref{Assumption for initial} hold.
      Let $\Omega$ be a bounded smooth convex domain, $\mu_0\in C(\Bar{\Omega})$ but we do not assume that $\mu_0 =0 $ on $\partial \Omega$. Additionally we suppose $\mu_0$ is $\frac{V}{2}$-relatively log-concave where $V(x)$ is the potential function of FP flow \eqref{FPflow}. Then the mutual information $I(X_0;X_t)$ is convex over $t$ perpetually.
  \end{lemma} 
  \begin{proof}
See Appendix \ref{appendix: proof of step2}.
  \end{proof}
  For an unbounded convex domain $\Omega$ in $\R^N$, we consider using a sequence of bounded smooth convex domains to converge to it.
  Then we deduce that the boundness of $\Omega$ is also redundant. Here we shall only consider the case of $\Omega=\R^N$ in the following discussion.
  
  %However, it is uncertain whether the kernel function corresponding to equation \eqref{schrodinger} on an arbitrary \(\Omega\) exhibits good decay properties, and whether its solution can be approximated by solutions to equations on a series of bounded convex domains. Therefore, we shall only consider the case of $\R^N$ in the following discussion.
  \begin{theorem}\label{step3}
      Let Assumption \ref{Assumption for potential},\ref{Assumption for initial} hold.
      Let $\mu_0\in C(\R^N)$. Additionally we suppose $\mu_0$ is $\frac{V}{2}$-relatively log-concave where $V(x)$ is the potential function of FP flow \eqref{FPflow}. Then the mutual information $I(X_0;X_t)$ is convex over $t$ perpetually.
  \end{theorem}
  \begin{proof}
See Appendix \ref{appendix: proof of step3}.
  \end{proof}
  In particular, if there exists a sufficiently large time, such that the distribution
  at this time is $\frac{V}{2}$-relatively log-concave, then mutual information preserves convexity after this time.
  \begin{corollary}
      \label{step4}
      Let Assumption \ref{Assumption for potential},\ref{Assumption for initial} hold.
      Let $\mu_T\in C(\R^N), T\geqslant0$. If $\mu_T$ is $\frac{V}{2}$-relatively log-concave where $V(x)$ is the potential function of FP flow \eqref{FPflow}. Then the mutual information $I(X_0;X_t)$ preserves convexity for $t>T$.
  \end{corollary}
  %\begin{corollary}
  %    Precisely one of the following statements holds with assumptions \ref{Assumption for potential},\ref{Assumption for initial}:
  %    \begin{enumerate}
       %   \item The mutual information $I(\mu_0;\mu_t)$ is convex over $t>0$ perpetually.
        %  \item The mutual information $I(\mu_0;\mu_t)$ is concave over $t>0$ perpetually.
         % \item There exists a time point $T>0$, such that the mutual information $I(\mu_0;\mu_t)$ is concave in $[0,T]$ while convex in $[T,+\infty)$.
      %\end{enumerate}
       %That is, the mutual information \( I(\mu_0; \mu_t) \) has at most one inflection point. Once it becomes convex in \( t \), its convexity remains unchanged thereafter. 
  
  %\end{corollary}
  %Intuitively, however, the second statement should not hold, which necessitates further research.
  
  %%%%%
  %%% Do not uncomment for double-blind review submission!
  % \section*{Acknowledgments}
  %
  % We are indebted to Michael Shell for maintaining and improving
  % \texttt{IEEEtran.cls}. 

  \section{Discussion and Future Work}
  \label{sec:conclusion}
  
  In this paper we have studied the convexity of mutual information along the Fokker-Planck ﬂow. We considered the gradient ﬂow interpretation of the Fokker-Planck equation in the space of probability measures, and derived formulae for the various derivatives of relative entropy and mutual information. We have shown that mutual information is perpetually convex under specific conditions on the initial distribution and the potential function. 
  These results generalize the behaviors seen in the heat ﬂow and OU flow\cite{wibisono2018convexity,wibisono2018convexityOU}.
  
  For simplicity in this paper we have treated only the case when the initial state is $\frac{V}{2}$-relatively log-concave and the potential function satisfies Assumption \ref{Assumption for potential}. Indeed, there is an interesting dichotomy in which we understand the intricate properties of the OU process since we have an explicit solution, whereas we know very little about the FP flow. 
  This paper demonstrates the convexity of mutual information along the FP flow using a method without reliance on explicit solutions. The primary approach is based on deriving the expression for the second derivative of mutual information through the application of optimal transport theory and gradient flow theory, as well as leveraging the research on the convexity of solutions to parabolic equations by Ishige et al\cite{ishige2020spatial}.
  
  Some interesting future directions are to weaken the assumptions and explore some new properties of mutual information. For example, would the mutual information be eventually convex if the initial state is not $\frac{V}{2}$-relatively log-concave? We might consider imposing constraints on the initial state, such as requiring it to be strongly log-concave, bounded, or possessing finite fourth moments and Fisher information. This necessitates further investigation into the properties of the solutions to the Fokker-Planck equation.

%%%%%%
%% To balance the columns at the last page of the paper use this
%% command somewhere at the top of the first column of the last page:
%%
% \enlargethispage{-5cm} 
%%
%% where the exact amount of page reduction has to be adapted to the
%% actual situation.
%%
%% If the balancing should occur in the middle of the references, use
%% the following trigger:
%%
% \IEEEtriggeratref{3}
%%
%% which triggers a \newpage (i.e., new column) just before the given
%% reference number. Note that you need to adapt this if you modify
%% the paper. The "triggered" command can be changed if desired:
%%
% \IEEEtriggercmd{\enlargethispage{-20cm}}
%%
%%%%%%

%%%%%%
%% References:
%% We recommend the usage of BibTeX:
%%
\bibliographystyle{IEEEtran}
\bibliography{references} % 假设你的BibTeX文件名为references.bib

%\bibliography{definitions,bibliofile}
%%
%% where we here have assume the existence of the files
%% definitions.bib and bibliofile.bib.
%% BibTeX documentation can be obtained at:
%% http://www.ctan.org/tex-archive/biblio/bibtex/contrib/doc/
%%%%%%
%% Or you use manual references (pay attention to consistency and the
%% formatting style!):

%%%%%% 
%% Appendix:
%% If needed a single appendix is created by
%%
\onecolumn
\appendices

%%
%% If several appendices are needed, then the command
%%
% \appendices
%%
%% in combination with further \section-commands can be used.
%%%%%%
\section{Some Applications of Stochastic Analysis in PDE}
\label{appendix: Some Applications of Stochastic Analysis in PDE}
In this subsection, we consider the applications of stochastic analysis in solving PDE. 
\begin{example}
    Consider the heat flow:
   $$ \begin{cases}
        \partial_t u = \Delta u, (t,x)\in \R_+\times \R^N\\
        u(0,x)=u_0(x), u_0\in C_b(\R^N)
    \end{cases}$$
    and we assume that $u\in C^{1,2}(\R_+\times\R^N)$. Define a stochastic process
    $Z_s:=\frac{1}{2}u(2(t-s),2W_s), 0\leqslant s\leqslant t$. By It$\hat{o}$'s formula, we have
\begin{equation*}
    \begin{aligned}
        dZ_s&=\dfrac{1}{2}\cdot(-2)\partial_t u ds+\dfrac{1}{2}\cdot2\partial_{i}u dW_s^i+\dfrac{1}{2}\cdot\dfrac{1}{2}\partial_{ij}u d\left<2W_s^i, 2W_s^j\right>\\
        &=-\partial_t u+\partial_i u dW_s^i+\Delta u ds\\
        &=0+\partial_i u dW_s^i.
    \end{aligned}
\end{equation*}
Hence $Z_s$ is a martingale. Here we use the Einstein summation convention, $\partial_i u$ is the partial derivative of $u$ with respect to $x_i$, and $\left<X,Y\right>$ is the quadratic covariation of two process $X,Y$. By the properties of martingales, we have
\begin{equation*}
    \begin{aligned}
        \E[Z_t]&=\E[Z_0]=\dfrac{1}{2}\E[u(2t,2W_0)]\\
        &=\dfrac{1}{2}u(2t,2W_0)
        =\dfrac{1}{2}u(2t,2x).
    \end{aligned}
\end{equation*}
Note that we always regard $W_0$ as the starting point $x$. In addition, this expectation can be represented by an integration
\begin{equation*}
    \begin{aligned}
        \E[Z_t]&=\dfrac{1}{2}\E[u(0,2W_t)]
        =\dfrac{1}{2}\E[u_0(2W_t)]\\
        &=\dfrac{1}{2}\int_{\R^N} u_0(2y)(2\pi t)^{-\frac{n}{2}}e^{-\frac{\norm{y-x}^2}{2t}} dy\\
        &\xlongequal{z=2y}
\dfrac{1}{2}\cdot \dfrac{1}{2^n}\int_{\mathbb{R} ^N}{u_0\left( z\right)}\left( 2\pi t \right) ^{-\frac{n}{2}}e^{-\frac{\left\| z/2-x \right\| ^2}{2t}}dz.
    \end{aligned}
\end{equation*}
Hence
\begin{equation}\label{solve heat flow}
    \begin{aligned}
        u(2t,2x)&=\int_{\R^N} u_0(2y)(2\pi t)^{-\frac{n}{2}}e^{-\frac{\norm{y-x}^2}{2t}} dy\\
        u(t,x)&=
\dfrac{1}{2^n}\int_{\mathbb{R} ^N}{u_0\left( y \right)}\left( \pi t \right) ^{-\frac{n}{2}}e^{-\frac{\left\| y/2-x/2 \right\| ^2}{t}}dy\\
&=
\int_{\mathbb{R} ^N}{u_0\left( y \right)}\left( 4\pi t \right) ^{-\frac{n}{2}}e^{-\frac{\left\| y-x \right\| ^2}{4t}}dy.
    \end{aligned}
\end{equation}
\end{example}
In equation (\ref{solve heat flow}), we have solved the heat equation with initial distribution. Then Lemma \ref{stochastic analysis_FP flow} can be proved by similar methods.
\begin{lemma}\label{stochastic analysis_FP flow}
    Consider this Schr\"odinger equation:
\begin{equation*}
    \begin{cases}
        \partial_t u=\Delta u -c(x)u, (t,x)\in\R_+\times \R^N\\
        u(0,x)=u_0(x), u_0\in C_b(\R^N)
    \end{cases}
\end{equation*}
and we assume that $u\in C^{1,2}(\R_+\times \R^N)$, that is, $u\in C^{1}(\R_+)$ over $t$ and $u\in C^{2}(\R^N)$ over $x$. Then the solution has the following representations:
\begin{equation*}
    \begin{aligned}
        u(t,x)&=\E_{x/2}\left[
        u_0(2W_{t/2})e^{-2\int_{0}^{t/2}c(2W_r)dr}
        \right]\\
        &=\E_{x/\sqrt{2}}\left[
        u_0(\sqrt{2}W_{t})e^{-\int_{0}^{t}c(\sqrt{2}W_r)dr}
        \right]
        \\
        &\equiv\int_{\R^N} K(x,y,t)u_0(y) dy
    \end{aligned}
\end{equation*}
\end{lemma}
\begin{proof}
   Define
\begin{equation*}
\begin{aligned}
        Z_s&:=u(2(t-s),2W_s)e^{-2\int_{0}^{s}c(2W_r) dr},  0\leqslant s\leqslant t\\
    Y_s&:=e^{-2\int_{0}^{s}c(2W_r) dr}, 0\leqslant s\leqslant t.
\end{aligned}
\end{equation*}
By It\^o's formula, we have
\begin{equation*}
    \begin{aligned}
        dZ_s&=Y_s\left(-2\partial_t u ds+2\partial_i u dW_s^i+\dfrac{1}{2}\partial_{ij}u d\left<2W_s^i, 2W_s^j\right>\right)
        + u\cdot \left[
-2Y_s\cdot c\left( 2W_s \right) \,\,ds + M_i(W_s^i)dW_s^i+\dfrac{1}{2}d\left<Y_s^i,Y_s^j\right>
\right]\\
&\;\;\;+d\left<u(2(t-s),2W_s), Y_s\right>\\
&\xlongequal{(i)}2Y_s\left(
-\partial _tu+\Delta u+c\left(2W_s\right)u
\right)ds+u\cdot M_i(W_s^i)dW_s^i+0\\
&=u\cdot M_i(W_s^i)dW_s^i
    \end{aligned}
\end{equation*}
where (\textit{i}) holds since $Y_s$ is a bounded variation process, and $M_i$ is the partial derivative of $Y_s$ with respect to $x_i$. Hence $Z_s$ is a martingale. By the properties of martingales,
\begin{equation*}
    \begin{aligned}
        \E[Z_t]&=\E[Z_0]=\E[u(2t,2W_0)]=u(2t,2x)\\
        \E[Z_t]&=\E_x\left[
        u(0,2W_t)e^{-2\int_{0}^{t}c(2W_r)dr}
        \right]\\
        &=\E_x\left[
        u_0(2W_t)e^{-2\int_{0}^{t}c(2W_r)dr}
        \right]\\
        u(t,x)&=\E_{x/2}\left[
        u_0(2W_{t/2})e^{-2\int_{0}^{t/2}c(2W_r)dr}
        \right]\\
        &\equiv\int_{\R^N} K(x,y,t)u_0(y) dy.
    \end{aligned}
\end{equation*}
In this context, the subscript of $\E$ represents the starting point of Brownian motion and 
$K(x,y,t)$ is a fixed kernel function. This kernel serves a role analogous to that of the heat kernel in the resolution of heat flow problems. In addition, if we take 
$$Z_t:=u(t-s, \sqrt{2}W_s)e^{-\int_{0}^{s} c(\sqrt{2}W_r)}dr$$
then by the same way we can obtain
\begin{equation*}
    \begin{aligned}
         u(t,x)&=\E_{x/\sqrt{2}}\left[
        u_0(\sqrt{2}W_{t})e^{-\int_{0}^{t}c(\sqrt{2}W_r)dr}
        \right].
    \end{aligned}
\end{equation*}
\end{proof}

\section{
Proof of Theorem \ref{Thm:Existence and Uniqueness to the Cauchy Problem}}
\label{appendix: Existence and Uniqueness to the Cauchy Problem}

In this context, we consider to prove the existence and uniqueness of a classical solution to a class of equations:
\begin{equation}
\label{Cauchy Problem in the entire space}
    \begin{aligned}
    \begin{cases}
         -\partial_t u +\Delta u - c(x)u=0, & (t,x)\in \R_+\times \R^N\\
         u(0,x)=u_0(x), & u_0\in C_b(\R^N)
    \end{cases}
    \end{aligned}
\end{equation}
where $c(x)$ is a nonnegative convex smooth function. Firstly we consider the Cauchy problem in a compact domain with smooth boundary:
\begin{equation}
\label{eq: Cauchy problem in a bounded domain}
    \begin{aligned}
    \begin{cases}
         -\partial_t u +\Delta u - c(x)u=0, & (t,x)\in \R_+\times \Omega\\
         u(0,x)=u_0(x), & \text{for }x \in \overline{\Omega}
    \end{cases}
    \end{aligned}
\end{equation}

By \textit{Theorem 10} in Chapter 1 of \cite{friedman2008partial}, it is established that the equation \eqref{eq: Cauchy problem in a bounded domain} has a unique classical solution and the kernel function $K_{\Omega}(x,y,t)$ in \eqref{kernel function in a bounded domain} is bounded in $\Omega$, which equals to the fundamental solution taken $\tau=0$. Then we prove the existence and uniqueness of a classical solution to \eqref{Cauchy Problem in the entire space} by utilizing the method in \cite{rubio2011existence}.
\begin{lemma}\textit{(Continuity of the Feynman-Kac solution)}
\label{Continuity of the Feynman-Kac solution}
    Let $\displaystyle v(t,x):=\E_{x/2}\left[
    u_0(2W_{t/2})e^{-2\int_{0}^{t/2}c(2W_r)dr}
        \right]$ which is define in the proof of Lemma \ref{stochastic analysis_FP flow}, then $v(t,x)$ is continuous in $[0,\infty)\times \R^N$.
\end{lemma}
\begin{proof}
   Here we define $\displaystyle W_x(r)\equiv W_r(x)$ as the Brownian motion starting from $x$ at time $r$. Then we rewrite $v(t,x)$ as 
   \begin{equation}
       \begin{aligned}
       \label{The Feynman-Kac solution}
           v(t,x)&=\E\left[
    u_0(2W_{x/2}(t/2))e^{-2\int_{0}^{t/2}c(2W_{x/2}(r))dr}
        \right].
       \end{aligned}
   \end{equation}
To simplify this function, we consider to prove the continuity of $\displaystyle v_1(t,x):=v(2t,2x)=\E\left[
    u_0(2W_{x}(t))e^{-2\int_{0}^{t}c(2W_{x}(r))dr}
        \right]$.
Let $\{\{(t_n,x_n)\}_{n\in\BN},(t,x)\}\subset (0,\infty)\times \R^N$. Assume that $(t_n,x_n)\xrightarrow[]{n\to\infty}(t,x)$. We need to prove that
$$v_1(t_n,x_n)\xrightarrow[]{n\to\infty}v_1(t,x)$$
that is, for any $\epsilon>0$, there exists an $N_1\in\BN$ such that for any $n\geqslant N_1$
$$|v_1(t_n,x_n)-v_1(t,x)|<\epsilon.$$
Let $\epsilon>0,0<\alpha \ll 1$ and $N_1\in \BN$
such that
$$\norm{(t_n,x_n)-(t,x)}<\alpha\text{ for }n\geqslant N_1.$$
Then if $n\geqslant N_1$ we get $t-\alpha<t_n<t+\alpha$ and $\norm{x}-\alpha<\norm{x_n}<\norm{x}+\alpha$. Note that $W_x(t)$ in $v_1(t,x)$ is the solution to this simplest stochastic differential equation:
\begin{equation}
\label{simplest SDE}
    \begin{aligned}
        &dX(t)=dW(t),\\
        &X(0)=x.
    \end{aligned}
\end{equation}
By \textit{Proposition 2.1} in \cite{rubio2011existence}, $W_x(t)\equiv W(t,x)$ is continuous a.s. in $(0,\infty)\times \R^N$. To prove the continuity, we define these random variables
\begin{equation*}
    \begin{aligned}
        Y_n:= e^{-2\int_{0}^{t} c(2W_x(r))dr}u_0(2W_x(t))- e^{-2\int_{0}^{t_n} c(2W_{x_n}(r))dr}u_0(2W_{x_n}(t_n)).
    \end{aligned}
\end{equation*}
These random variables are uniformly integrable since
\begin{equation}
    \begin{aligned}
        \int_{\R^N} Y_n^2 d\pr
        &\leqslant 2 \int_{\R^N} \left(
        e^{-2\int_{0}^{t} c(2W_x(r))dr}u_0(2W_x(t))
        \right)^2 d\pr+
        2 \int_{\R^N} \left(
       e^{-2\int_{0}^{t_n} c(2W_{x_n}(r))dr}u_0(2W_{x_n}(t_n))
        \right)^2 d\pr\\
        &\leqslant 4\sup\limits_{x\in\R^N} |u_0(x)| \int_{\R^N} d\pr\\
        &= 4\sup\limits_{x\in\R^N} |u_0(x)|<\infty.
    \end{aligned}
\end{equation}
It follows from \textit{Theorem 4.2 page 215} in \cite{gut2006probability} that $\{Y_n\}_{n\geqslant N_1}$ is uniformly integrable. Let $0<\eta<1$ and $M>0$. Define $M_1:=1+M, A_1:=[0,t+\alpha]\times [-M_1,M_1]^N, A_2:= [-M_1,M_1]^N$ and $\norm{\nu}_T:=\sup\limits_{0\leqslant s\leqslant T}\norm{\nu(s)}$ for a function $\nu: [0,\infty)\to\R^N, T>0$. As a consequence of uniformly integrable of $\{Y_n\}_{n\geqslant N_1}$, for any $\epsilon>0$ there exists a $\delta(\epsilon)>0$, such that for any measurable set $B$ with $\pr[B]<\delta(\epsilon)$, we have $\sup_n\int_{B}|Y_n|d\pr<\frac{\epsilon}{2}$. By the continuity of $W(t,x)$ and $c(x)$ and the property of not exploding in finite time a.s., we may choose $M>0$ such that 
$$\pr[\norm{W_x}_{t+\alpha}>M]\leqslant \dfrac{\delta(\epsilon)}{2}$$
and $N_2\in\BN$ for which 
\begin{equation*}
    \begin{aligned}
        \pr[\norm{W_{x_n}-W_{x}}_{t+\alpha}>\eta]\leqslant \dfrac{\delta(\epsilon)}{2},\quad 
        \pr[\norm{c(W_{x_n})-c(W_{x})}_{t+\alpha}>\eta]\leqslant \dfrac{\delta(\epsilon)}{2}.
    \end{aligned}
\end{equation*}
Let $B_1:=\{\norm{W_x}_{t+\alpha}\leqslant M\},B_2:=\{\norm{W_{x_n}-W_{x}}_{t+\alpha} \lor \norm{c(W_{x_n})-c(W_{x})}_{t+\alpha}\leqslant\eta\}$, where $A\lor B:=\max \{A,B\}, A\land B :=\min \{A,B\}$.
Then
\begin{equation}
    \begin{aligned}
        |v_1(t_n,x_n)-v_1(t,x)|&\leqslant
        \int_{B_1\cap B_2}|Y_n| d\pr+
    \int_{\R^N\backslash (B_1\cap B_2)}|Y_n| d\pr\\
    &\leqslant 
    \int_{B_1\cap B_2}|Y_n| d\pr+ \dfrac{\epsilon}{2}\\
    &=\int_{B_1\cap B_2} 
    \left|e^{-2\int_{0}^{t} c(2W_x(r))dr}u_0(2W_x(t))- e^{-2\int_{0}^{t_n} c(2W_{x_n}(r))dr}u_0(2W_{x_n}(t_n))\right| d\pr +\dfrac{\epsilon}{2}\\
    &\leqslant 
    \int_{B_1\cap B_2} e^{-2\int_{0}^{t} c(2W_x(r))dr}
    \left|u_0(2W_x(t))-u_0(2W_{x_n}(t_n))\right| d\pr\\
    &\;\;  +\int_{B_1\cap B_2} |u_0(2W_{x_n}(t_n))| 
    \left|
    e^{-2\int_{0}^{t} c(2W_x(r))dr}-
    e^{-2\int_{0}^{t_n} c(2W_{x_n}(r))dr}
    \right| d\pr+
\dfrac{\epsilon}{2}\\    &:=I_n+J_n+\dfrac{\epsilon}{2}.
    \end{aligned}
\end{equation}
By the Lebesgue dominated convergence theorem (LDCT) and nonnegativity of $c(x)$, we can find $N_3\in \BN$ such that for any $n\geqslant N_3$, we have $I_n<\frac{\epsilon}{4}$. We then estimate $J_n$:
\begin{equation*}
    \begin{aligned}
        J_n&\leqslant
        \sup\limits_{x\in\R^N}|u_0(x)| \int_{B_1\cap B_2}\left|
        e^{2\int_{0}^{t_n}c(2W_{x_n}(r))dr-2\int_{0}^{t}c(2W_{x}(r))dr}-1
        \right| d\pr.
    \end{aligned}
\end{equation*}
Since in $B_1\cap B_2$, 
\begin{equation*}
    \begin{aligned}
        q_n&:=\int_{0}^{t_n} c(2W_{x_n}(r)) dr-\int_{0}^{t}c(2W_x(r))dr\\
        &= \int_{0}^{t_n} c(2W_{x_n}(r)) dr-
        \int_{0}^{t}c(2W_{x_n}(r))dr+
        \int_{0}^{t}c(2W_{x_n}(r))dr-\int_{0}^{t}c(2W_x(r))dr\\
        &\leqslant \int_{t_n\lor t}^{t_n\land t} c(2W_{x_n}(r))dr+ \int_{0}^{t}c(2W_{x_n}(r))dr-\int_{0}^{t}c(2W_x(r))dr\\
        &\leqslant 2\alpha \sup\limits_{x\in A_2}c(x) + \int_{0}^{t} \left(
        c(2W_{x_n}(r))-c(2W_x(r))
        \right)dr.
    \end{aligned}
\end{equation*}
Then $q_n\xrightarrow[]{n\to\infty} 0$ and there exists an $N_4\in\BN$ such that $J_n<\frac{\epsilon}{4}$ by LDCT. Hence we prove that for any $n\geqslant \max\limits_{i=1,2,3,4}\{N_i\}$, we have $|v_1(t_n,x_n)-v_1(t,x)|<\epsilon$. Then $v_1(t,x)$ and $v(t,x)$ is continuous in $(0,\infty)\times \R^N$. For the continuity at $t = 0$ we can proceed the same way.  
\end{proof}
\begin{lemma}
    \label{Differentiability of the Feynman-Kac solution}\textit{(Differentiability of the Feynman-Kac solution)} Let $v$ defined as in equation \eqref{The Feynman-Kac solution}. Then $v\in C^{1,\infty}((0,\infty)\times \R^N)$.
\end{lemma}
\begin{proof}
    Let $T>0$. Consider the following parabolic differential equation in an open bounded domain $A\subset \R^N$ with $C^2$ boundary:
    \begin{equation}
    \label{eq: pde of parabolic type in a bounded domain}
        \begin{aligned}
            \begin{cases}
                -\partial_t u+\Delta u -c(x) u =0, &(t,x)\in [0,T]\times A\\
                u(0,x)=v(0,x), &x\in A\\
                u(t,x)=v(t,x), & (t,x)\in (0,T]\times A
            \end{cases}
        \end{aligned}
    \end{equation}
From the continuity of $v$, the existence and uniqueness of a classical solution to equation \eqref{eq: pde of parabolic type in a bounded domain} follows from \cite{friedman2008partial}. Define the following stopping time
$$\tau:=\inf \{s>0| W_x(s)\notin \overline{A}\}.$$
Following the same arguments of \textit{Theorem 2.3} in \cite{rubio2011existence} and Section 5 in Chapter 6 of \cite{friedman1975stochastic}, we can prove that the classical solution to \eqref{eq: pde of parabolic type in a bounded domain} has the following representation:
\begin{equation}
    \begin{aligned}
        w(t,x)=\E_{x/\sqrt{2}}\left[
        v(t-\tau, \sqrt{2}W(\tau)) e^{-\int_{0}^{\tau} c(\sqrt{2}W(r))dr})
        \right].
    \end{aligned}
\end{equation}
Using the strong Markov property of the process $W_x(s)$, we prove the equality between $v$ and $w$. Consider the filtration $\mathcal{F}_{\tau}$. For $v(t,x)$ we use the property of conditional expectation to get:
\begin{equation}
    \begin{aligned}
        v(t,\sqrt{2}x)&=\E_x \left[
        u_0(\sqrt{2}W(t)) e^{-\int_{0}^{t} c(\sqrt{2}W(r))dr}
        \right]\\
        &=\E_x\left[
        \E\left[\left.
         u_0(\sqrt{2}W(t)) e^{-\int_{0}^{t} c(\sqrt{2}W(r))dr}\right|\mathcal{F}_\tau
        \right]
        \right]\\
        &=\E_x\left[e^{-\int_{0}^{\tau}c(\sqrt{2}W(r))dr}
        \E\left[\left.
         e^{-\int_{0}^{t-\tau} c(\sqrt{2}W(r+\tau))dr}u_0(\sqrt{2}W(t-\tau+\tau))\right|\mathcal{F}_\tau
        \right]
        \right]\\
        &=\E_x\left[e^{-\int_{0}^{\tau}c(\sqrt{2}W(r))dr}
        \E_{\sqrt{2}W(\tau)}\left[
         e^{-\int_{0}^{t-\tau} c(\sqrt{2}W(r))dr}u_0(\sqrt{2}W(t-\tau))
        \right]
        \right]\\
        &=\E_x\left[e^{-\int_{0}^{\tau}c(\sqrt{2}W(r))dr}
        v(t-\tau, \sqrt{2}W(\tau))
        \right]\\
        &=w(t,\sqrt{2}x).
    \end{aligned}
\end{equation}
Hence $v(t,x)\in C^{1,2}([0,T]\times A)$ and it satisfies the parabolic equation:
$\partial_t v=\Delta v-c(x)v$. Since $T>0$ and the set $A$ are arbitrary, we get that $v(t,x)\in C^{1,2}(\R_+\times \R^N)$. Furthermore, by the expression of heat equations we have:
$$u(x,t)=\dfrac{1}{(2\sqrt{\pi t})^N}\int_{\R^N} u_0(y) e^{-\frac{\norm{y-x}^2}{4t}}dy-\dfrac{1}{(2\sqrt{\pi})^N}\int_{0}^{t}\int_{\R^N}\dfrac{c(y)u(y,\tau)}{(\sqrt{t-\tau})^N}e^{-\frac{\norm{y-x}^2}{4(t-\tau)}}dyd\tau$$
which implies that $u\in C^{1,\infty}(\R_+\times \R^N)$.
\end{proof}

By combining Lemma \ref{stochastic analysis_FP flow}, Lemma \ref{Continuity of the Feynman-Kac solution} and Lemma \ref{Differentiability of the Feynman-Kac solution}, we can prove the uniqueness and existence of solutions to Fokker-Planck flow \eqref{FPflow} and the Cauchy problem \eqref{schrodinger2}.

\section{Proof of Lemma \ref{derivatives of entropy}}
\label{appendix: derivatives of entropy}
According to the theory of gradient flow and optimal transport, we can interpret the FP flow as the gradient flow $\dot{\mu}=-\grad_{\mu}H_{\nu}$ of relative entropy
$$H_{\nu}(\mu)=\int_{\R^N}\mu\log\dfrac{\mu}{\nu} dx$$
where the reference measure $\nu=\mu_{\infty}=e^{-V}$. Then we have
\begin{equation*}
    \begin{aligned}
        \dfrac{d}{dt}H_{\nu}(\mu_t)&=\left<\grad_{\mu}H_{\nu},\dfrac{\partial \mu_t}{\partial t}\right>\\
        &=-\norm{\grad_{\mu} H_{\nu}}^2\\
        &=-\int_{\R^N} \mu_t\norm{\nabla \log\dfrac{\mu_t}{\nu}}^2 dx\\
        &=-J_{\nu}(\mu_t).
    \end{aligned}
\end{equation*}

Similarly, the second derivative of relative entropy along the gradient flow can be given by the Hessian operator:
\begin{equation*}
    \begin{aligned}
        \dfrac{d^2}{dt^2}H_{\nu}(\mu_t)&=
        -\dfrac{d}{dt}\norm{\grad_{\mu}H_{\nu}}^2
        \\
        &=
        2(\Hess_{\mu}H_{\nu})(\grad_{\mu}H_{\nu}).
    \end{aligned}
\end{equation*}
By the Formula 15.7 or Theorem 24.2 in \cite{villani2009optimal}, this second derivative becomes the following equation
\begin{equation*}
    \begin{aligned}
        \frac{d^2}{dt^2}H_{\nu}\left( \mu _t \right)&=2K_{\nu}\left( \mu _t \right) +2G_{\nu}(\mu_t)
    \end{aligned}
\end{equation*}
where
\begin{equation*}
    \begin{aligned}
       K_{\nu}(\mu)&=\int_{\R^N}\mu \HS{\nabla^2 \log\dfrac{\mu}{\nu}}^2 dx\\
       G_{\nu}(\mu)&=
\int_{\mathbb{R} ^N}{\mu}\left< \left( \nabla ^2V \right) \nabla \log \frac{\mu}{\nu}, \nabla \log \frac{\mu}{\nu} \right> dx.
    \end{aligned}
\end{equation*}
In particular, if $\mu_0$ is a point mass at $x_0\in\R^N$, we have
\begin{equation*}
    \begin{aligned}
        \dfrac{d}{dt}H_{\nu}(\mu_t|\mu_0=x_0)
        &=-J_{\nu}(\mu_t|\mu_0=x_0)\\
         \frac{d^2}{dt^2}H_{\nu}\left( \mu _t |\mu_0=x_0\right)&=2K_{\nu}\left( \mu _t|\mu_0=x_0 \right) +2G_{\nu}(\mu_t|\mu_0=x_0).
    \end{aligned}
\end{equation*}

Hence for any initial distribution $\mu_0$, we take the integration over $\mu_0$ to obtain
\begin{equation*}
    \begin{aligned}
        \dfrac{d}{dt}H_{\nu}(\mu_t|\mu_0)
        &=-J_{\nu}(\mu_t|\mu_0)\\
         \frac{d^2}{dt^2}H_{\nu}\left( \mu _t |\mu_0\right)&=2K_{\nu}\left( \mu _t|\mu_0 \right) +2G_{\nu}(\mu_t|\mu_0).
    \end{aligned}
\end{equation*}

\section{Proof of Theorem \ref{derivatives of mutual information}}
\label{appendix: derivatives of mutual information}
The proof of Theorem \ref{derivatives of mutual information} is based on Lemma \ref{expansion of Kv} and \ref{integration by parts}.
\begin{lemma}\label{expansion of Kv}
    For any joint distribution $(\mu_0,\mu_t)$ and any reference probability measure $\nu$,
    \begin{equation*}
        \begin{aligned}
            K_{\nu}&(\mu_0;\mu_t)=
\Psi \left( \mu _0|\mu _t \right) +
2\int_{\mathbb{R} _{t}^{N}}{\mu _t\left< \nabla ^2\log \frac{\mu _t}{\nu}, \int_{\mathbb{R} _{0}^{N}}{\mu _{0|t}}\nabla ^2\log \,\mu _{0|t}dx_0 \right> _{HS}}dx_t
        \end{aligned}
    \end{equation*}
\end{lemma}
\begin{proof}
    We have the following decomposition:
\begin{equation*}
    \begin{aligned}
\,\,\nabla ^2\log \frac{\mu _{t|0}}{\nu}=\nabla ^2\log \frac{\mu _t}{\nu}+\nabla ^2\log \,\mu _{0|t}.
    \end{aligned}
\end{equation*}
Then we obtain
\begin{equation*}
    \begin{aligned}
K_{\nu}\left( \,\mu _0;\mu _t \right) &=K_{\nu}\left( \mu _t|\mu _0 \right) -K_{\nu}\left( \mu _t \right) \\
&=
\int_{\mathbb{R} _{0}^{N}}{\int_{\mathbb{R} _{t}^{N}}{\mu _{t,0}}\left( \left\| \nabla ^2\log \frac{\mu _{t|0}}{\nu} \right\| _{HS}^{2}-\left\| \nabla ^2\log \frac{\mu _t}{\nu} \right\| _{HS}^{2} \right)}\,\,dx_tdx_0\\
&=
\int_{\mathbb{R} _{0}^{N}}{\int_{\mathbb{R} _{t}^{u}}{\mu _{t,0}}\left\| \nabla ^2\log \frac{\mu _{0|t}}{\nu} \right\| _{HS}^{2}}\,\,dx_tdx_0 +2\int_{\mathbb{R} _{0}^{N}}{\int_{\mathbb{R} _{t}^{u}}{\mu _{t,0}}\left< \nabla ^2\log \frac{\mu _t}{\nu}, \nabla ^2\log \,\mu _{0|t} \right> _{HS}}\,\,dx_tdx_0\\
&=
\Psi \left( \mu _0|\mu _t \right)  +2\int_{\mathbb{R} _{t}^{N}}{\mu _t\int_{\mathbb{R} _{0}^{N}}{\mu _{0|t}}\left< \nabla ^2\log \frac{\mu _t}{\nu}, \nabla ^2\log \,\mu _{0|t} \right> _{HS}}\,\,dx_tdx_0\\
&=
\Psi \left( \mu _0|\mu _t \right) +2\int_{\mathbb{R} _{t}^{N}}{\mu _t\left< \nabla ^2\log \frac{\mu _t}{\nu}, \int_{\mathbb{R} _{0}^{N}}{\mu _{0|t}}\nabla ^2\log \,\mu _{0|t}dx_0 \right> _{HS}}\,\,dx_t
    \end{aligned}
\end{equation*}
where $$
\Psi \left( \mu _0|\mu _t \right) :=\int_{\mathbb{R} _{0}^{N}}{\int_{\mathbb{R} _{t}^{N}}{\mu _{t,0}}\left\| \nabla ^2\log \frac{\mu _{0|t}}{\nu} \right\| _{HS}^{2}}dx_tdx_0.
$$
\end{proof}

\begin{lemma}\label{integration by parts}
    For any distribution $\rho_t$ which is independent of $x_0$, any joint distribution $(\mu_0,\mu_t)$ and any reference probability measure $\nu$, we have the following formula of integration by parts:
    \begin{equation*}
        \begin{aligned}
\int_{\mathbb{R} _{t}^{N}}\mu _t&\int_{\mathbb{R} _{0}^{N}}{\mu _{0|t}}\left< \rho_t , \nabla \log \mu _{0|t}\left( \nabla \log \mu _{0|t} \right) ^{\top} \right> _{HS}\,\,dx_tdx_0
=
-\int_{\mathbb{R} _{t}^{N}}{\mu _t}\int_{\mathbb{R} _{0}^{N}}{\mu _{0|t}}\left< \rho_t , \nabla ^2\log \mu _{0|t} \right> _{HS}dx_tdx_0.
        \end{aligned}
    \end{equation*}
\end{lemma}
\begin{proof}

    \begin{equation*}
        \begin{aligned}
        \int_{\mathbb{R} _{t}^{N}}\mu _t\int_{\mathbb{R} _{0}^{N}}\mu _{0|t}&\left< \rho_t , \nabla \log \mu _{0|t}\left( \nabla \log \mu _{0|t} \right) ^{\top} \right> _{HS}\,\,dx_tdx_0\\
&= 
\int_{\mathbb{R} _{t}^{N}}{\mu _t}\left< \rho_t , \int_{\mathbb{R} _{0}^{N}}{\mu _{0|t}\nabla \log \mu _{0|t}\left( \nabla \log \mu _{0|t} \right) ^{\top}}dx_0 \right> _{HS}dx_t\\
&=
\int_{\mathbb{R} _{t}^{N}}{\mu _t}\left< \rho_t , \int_{\mathbb{R} _{0}^{N}}{\mu _{0|t}\left( \frac{\nabla ^2\mu _{0|t}}{\mu _{0|t}}-\nabla ^2\log \mu _{0|t} \right)}dx_0 \right> _{HS}dx_t
\\
&=
\int_{\mathbb{R} _{t}^{N}}{\mu _t}\left< \rho_t, \nabla ^2\int_{\mathbb{R} _{0}^{N}}{\mu _{0|t}}dx_0 \right> _{HS}dx_t + 
\int_{\mathbb{R} _{t}^{N}}{\mu _t}\int_{\mathbb{R} _{0}^{N}}{\mu _{0|t}}\left< \rho_t , -\nabla ^2\log \mu _{0|t} \right> _{HS}dx_tdx_0\\
&=
0+\int_{\mathbb{R} _{t}^{N}}{\mu _t}\int_{\mathbb{R} _{0}^{N}}{\mu _{0|t}}\left< \rho_t , -\nabla ^2\log \mu _{0|t} \right> _{HS}dx_tdx_0\\
&=
-\int_{\mathbb{R} _{t}^{N}}{\mu _t}\int_{\mathbb{R} _{0}^{N}}{\mu _{0|t}}\left< \rho_t , \nabla ^2\log \mu _{0|t} \right> _{HS}dx_tdx_0.
        \end{aligned}
    \end{equation*}

\end{proof}
By combining Lemma \ref{derivatives of entropy}, \ref{expansion of Kv} and \ref{integration by parts}, we can obtain the first and the second derivative of mutual information over $t$
\newpage
\begin{equation*}
    \begin{aligned}
    \dfrac{d}{dt}&I(\mu_0;\mu_t)=
       \dfrac{d}{dt}H_{\nu}(\mu_{t|0})- \dfrac{d}{dt}H_{\nu}(\mu_t)
        =J_{\nu}(\mu_{t|0})-J_{\nu}(\mu_t)\\
        &\qquad\;\;\;\;\;\,\,  =-J_{\nu}(\mu_0;\mu_t)\leqslant 0\\
\dfrac{d^2}{dt^2}&I\left( \mu _0;\mu _t \right) =\dfrac{d^2}{dt^2}H_{\nu}\left( \mu _t|\mu _0 \right)-\dfrac{d^2}{dt^2}H_{\nu}\left( \mu _t \right)\\
&=
2\Psi \left( \mu _0|\mu _t \right) +4\int_{\mathbb{R} _{t}^{N}}{\mu _t\int_{\mathbb{R} _{0}^{N}}{\mu _{0|t}}\left< -\nabla ^2\log \frac{\mu _t}{\nu}, -\nabla ^2\log \,\mu _{0|t} \right> _{HS}}\,\,dx_tdx_0\\
&\;\;+2\int_{\mathbb{R} _{0}^{N}}{\mu _0dx_0\int_{\mathbb{R} _{t}^{N}}{\mu _{t|0}}}\left< \nabla ^2V,\nabla \log \frac{\mu _{t|0}}{\nu}\left( \nabla \log \frac{\mu _{t|0}}{\nu} \right) ^{\top} \right> _{HS}dx_t\\
&\;\;-2\int_{\mathbb{R} ^N}{\mu _t}\left< \nabla ^2V,\nabla \log \frac{\mu _t}{\nu}\left( \nabla \log \frac{\mu _t}{\nu} \right) ^{\top} \right> _{HS}dx_t\\
&=
2\Psi \left( \mu _0|\mu _t \right) +4\int_{\mathbb{R} _{t}^{N}}{\mu _t\int_{\mathbb{R} _{0}^{N}}{\mu _{0|t}}\left< -\nabla ^2\log \frac{\mu _t}{\nu}, \nabla \log \mu _{0|t}\left( \nabla \log \mu _{0|t} \right) ^{\top} \right> _{HS}}\,\,dx_tdx_0\\
&\;\;
+2\int_{\mathbb{R} _{0}^{N}}{\mu _0dx_0\int_{\mathbb{R} _{t}^{N}}{\mu _{t|0}}}\left< \nabla ^2V,\nabla \log \frac{\mu _{t|0}}{\nu}\left( \nabla \log \frac{\mu _{t|0}}{\nu} \right) ^{\top} \right> _{HS}dx_t\\
&\;\;
-2\int_{\mathbb{R} ^N}{\mu _t}\left< \nabla ^2V,\nabla \log \frac{\mu _t}{\nu}\left( \nabla \log \frac{\mu _t}{\nu} \right) ^{\top} \right> _{HS}dx_t\\
&=
2\Psi \left( \mu _0|\mu _t \right) +4\int_{\mathbb{R} _{t}^{N}}{\mu _t\int_{\mathbb{R} _{0}^{N}}{\mu _{0|t}}\left< -\nabla ^2\log \mu _t, \nabla \log \mu _{0|t}\left( \nabla \log \mu _{0|t} \right) ^{\top} \right> _{HS}}\,\,dx_tdx_0\\
&\;\;
+2\int_{\mathbb{R} _{0}^{N}}{\mu _0dx_0\int_{\mathbb{R} _{t}^{N}}{\mu _{t|0}}}\left< \nabla ^2V,\nabla \log \frac{\mu _{t|0}}{\nu}\left( \nabla \log \frac{\mu _{t|0}}{\nu} \right) ^{\top} \right> _{HS}dx_t\\
&\;\;
-2\int_{\mathbb{R} ^N}{\mu _t}\left< \nabla ^2V, \nabla \log \frac{\mu _t}{\nu}\left( \nabla \log \frac{\mu _t}{\nu} \right) ^{\top} \right> _{HS}dx_t\\
&\;\;
-4\int_{\mathbb{R} _{t}^{N}}{\mu _t\int_{\mathbb{R} _{0}^{N}}{\mu _{0|t}}\left< \nabla ^2V, \nabla \log \mu _{0|t}\left( \nabla \log \mu _{0|t} \right) ^{\top} \right> _{HS}}\,\,dx_tdx_0\\
&=
2\Psi \left( \mu _0|\mu _t \right) +4\int_{\mathbb{R} _{t}^{N}}{\mu _t}\int_{\mathbb{R} _{0}^{N}}{\mu _{0|t}}\left< \nabla ^2\log \mu _t, \nabla ^2\log \mu _{0|t} \right> _{HS}dx_tdx_0\\
&\;\;
+2\int_{\mathbb{R} _{0}^{N}}{\mu _0dx_0\int_{\mathbb{R} _{t}^{N}}{\mu _{t|0}}}\left< \nabla ^2V,\nabla \log \mu _{t|0}\left( \nabla \log \mu _{t|0} \right) ^{\top}+\nabla \log \nu \left( \nabla \log \nu \right) ^{\top}\right.
\left.-2\nabla \log \mu _{t|0}\left( \nabla \log \nu \right) ^{\top} \right> _{HS}dx_t\\
&\;\;
-2\int_{\mathbb{R} ^N}{\mu _t}\left< \nabla ^2V,\nabla \log \mu _t\left( \nabla \log \mu _t \right) ^{\top}+\nabla \log \nu \left( \nabla \log \nu \right) ^{\top}-2\nabla \log \mu _t\left( \nabla \log \nu \right) ^{\top} \right> _{HS}dx_t\\
&\;\;
-4\int_{\mathbb{R} _{t}^{N}}{\mu _t\int_{\mathbb{R} _{0}^{N}}{\mu _{0|t}}\left< \nabla ^2V, \nabla \log \mu _{0|t}\left( \nabla \log \mu _{0|t} \right) ^{\top} \right> _{HS}}\,\,dx_tdx_0\\
&=
2\Psi \left( \mu _0|\mu _t \right) +4\int_{\mathbb{R} _{t}^{N}}{\mu _t}\int_{\mathbb{R} _{0}^{N}}{\mu _{0|t}}\left< \nabla ^2\log \mu _t, \nabla ^2\log \mu _{0|t} \right> _{HS}dx_tdx_0\\
&\;\;-2\int_{\mathbb{R} _{t}^{N}}{\mu _0}\int_{\mathbb{R} _{0}^{N}}{\mu _{t|0}}\left< \nabla ^2V, \nabla ^2\log \mu _{t|0} \right> _{HS}dx_tdx_0
+2\int_{\mathbb{R} _{t}^{N}}{\mu _0}\int_{\mathbb{R} _{0}^{N}}{\mu _{t|0}}\left< \nabla ^2V, \nabla V\left( \nabla V \right) ^{\top} \right> _{HS}dx_tdx_0
\\
&\;\;
+4\int_{\mathbb{R} _{t}^{N}}{\mu _0}\int_{\mathbb{R} _{0}^{N}}{\mu _{t|0}}\left< \nabla ^2V, \nabla \log \mu _{t|0}\left( \nabla V \right) ^{\top} \right> _{HS}dx_tdx_0+2\int_{\mathbb{R} _{t}^{N}}{\mu _t}\left< \nabla ^2V, \nabla ^2\log \mu _t \right> _{HS}dx_t\\
&\;\;
-2\int_{\mathbb{R} _{t}^{N}}{\mu _t}\left< \nabla ^2V, \nabla V\left( \nabla V \right) ^{\top} \right> _{HS}dx_t
-4\int_{\mathbb{R} _{t}^{N}}{\mu _t}\left< \nabla ^2V, \nabla \log \mu _t\left( \nabla V \right) ^{\top} \right> _{HS}dx_t\\
&\;\;+4\int_{\mathbb{R} _{t}^{N}}{\mu _t\int_{\mathbb{R} _{0}^{N}}{\mu _{0|t}}\left< \nabla ^2V, \nabla ^2\log \mu _{0|t} \right> _{HS}}\,\,dx_tdx_0.
    \end{aligned}
\end{equation*}

Further simplify the formula, we obtain
\begin{equation*}
    \begin{aligned}
       \dfrac{d^2}{dt^2} &I(\mu_0;\mu_t)\\
       &=2\Psi(\mu_0|\mu_t)+
4\int_{\mathbb{R} _{0}^{N}\times \mathbb{R} _{t}^{N}}{\mu _{0,t}} \left<
\nabla^2\log\mu_t,\nabla^2\log\mu_{0|t}
\right>_{HS} dx_0dx_t+
4\int_{\mathbb{R} _{0}^{N}\times \mathbb{R} _{t}^{N}}{\mu _{0,t}} \left<
\nabla^2 V,\nabla^2\log\mu_{0|t}(\nabla V)^\top
\right>_{HS} dx_0dx_t\\
&\;\;
-2\int_{\mathbb{R} _{0}^{N}\times \mathbb{R} _{t}^{N}}{\mu _{0,t}} \left<
\nabla^2 V,\nabla^2\log\mu_{0|t}
\right>_{HS} dx_0dx_t+
4\int_{\mathbb{R} _{0}^{N}\times \mathbb{R} _{t}^{N}}{\mu _{0,t}} \left<
\nabla^2 V,\nabla^2\log\mu_{0|t}
\right>_{HS} dx_0dx_t\\
&=
2\Psi \left( \mu _0|\mu _t \right) +2\int_{\mathbb{R} _{0}^{N}\times \mathbb{R} _{t}^{N}}{\mu _{0,t}}\left< 2\nabla ^2\log \mu _t+{\nabla ^2V}, \nabla ^2\log \mu _{0|t} \right> _{HS}dx_0dx_t\\
&\;\;
+4\int_{\mathbb{R} _{t}^{N}}{\mu _t}\left< \nabla ^2V,\left( \nabla V \right) ^{\top}\int_{\mathbb{R} _{0}^{N}}{\mu _{0|t}}\nabla \log \mu _{0|t}dx_0 \right> _{HS}dx_t
    \end{aligned}
\end{equation*}
where
\begin{equation*}
    \begin{aligned}
\int_{\mathbb{R} _{t}^{N}}{\mu _t}&\left< \nabla ^2V,\left( \nabla V \right) ^{\top}\int_{\mathbb{R} _{0}^{N}}{\mu _{0|t}}\nabla \log \mu _{0|t}dx_0 \right> _{HS}dx_t\\
&\qquad\qquad=
\int_{\mathbb{R} _{t}^{N}}{\mu _t}\left< \nabla ^2V,\left( \nabla V \right) ^{\top}\int_{\mathbb{R} _{0}^{N}}{\nabla \mu _{0|t}dx_0} \right> _{HS}dx_t\\
&\qquad\qquad=\int_{\mathbb{R} _{t}^{N}}{\mu _t}\left< \nabla ^2V,\left( \nabla V \right) ^{\top}\nabla \int_{\mathbb{R} _{0}^{N}}{\mu _{0|t}dx_0} \right> _{HS}dx_t\\
&\qquad\qquad=0.
   \end{aligned}
\end{equation*}
Finally, we obtain the second derivative of mutual information over $t$
\begin{equation*}
    \begin{aligned}
        \dfrac{d^2}{dt^2}I(\mu_0;\mu_t)
        &=2\Psi \left( \mu _0|\mu _t \right) +2\int_{\mathbb{R} _{0}^{N}\times \mathbb{R} _{t}^{N}}{\mu _{0,t}}\left< 2\nabla ^2\log \mu _t+{\nabla ^2V}, \nabla ^2\log \mu _{0|t} \right> _{HS}dx_0dx_t\\
        &=2\Psi \left( \mu _0|\mu _t \right) +4\int_{\mathbb{R} _{0}^{N}\times \mathbb{R} _{t}^{N}}{\mu _{0,t}}\left< \nabla^2\log\gamma_t, \nabla ^2\log \mu _{0|t} \right> _{HS}dx_0dx_t\\
        &=2\Psi \left( \mu _0|\mu _t \right) -4\int_{\mathbb{R} _{0}^{N}\times \mathbb{R} _{t}^{N}}{\mu _{0,t}} 
        \left(\nabla\log\mu_{0|t}\right)^\top
        \left(\nabla^2\log\gamma_t\right)
        \left(\nabla\log\mu_{0|t}\right)
        dx_0dx_t.
    \end{aligned}
\end{equation*}

Here $\gamma_t := e^{\frac{V}{2}} \mu_t$. Therefore, we consider proving that $\gamma_t$ is log-concave for every $t\geqslant 0$. Consequently, the mutual information $I(\mu_0; \mu_t)$ will be shown to be convex over $t$.

\section{Proof of Lemma \ref{step1}}
\label{appendix: proof of step1}
 Since $\gamma_t = e^{\frac{V}{2}}\mu_t$ satisfies the Schr\"odinger equation:
 \begin{equation*}
     \begin{aligned}
         \dfrac{\partial\gamma_t}{\partial t}+\left(\dfrac{1}{4}\norm{\nabla V}^2-\dfrac{1}{2}\Delta V\right)\gamma_t-\Delta \gamma_t=0
     \end{aligned}
 \end{equation*}
where $c(x):=\dfrac{1}{4}\norm{\nabla V}^2-\dfrac{1}{2}\Delta V\geqslant 0$. By Corollary \ref{Schr\"odinger equation_log concave}, if $\gamma_0$ is log-concave in $\Omega$ with zero boundary value, that is, $X_0\sim\mu_0$ is $\frac{V}{2}$-relatively log-concave in $\Omega$, then $\gamma_t$ is perpetually log-concave in $\Omega$ for every $t>0$. Hence, the mutual information is convex over $t$.
\section{Proof of Lemma \ref{step2}}
\label{appendix: proof of step2}
The proof of Lemma \ref{step2} is based on Lemma \ref{log-concavity of heat flow}$-$\ref{find a log-concave convergent sequence to u_0}.
\begin{lemma}\label{log-concavity of heat flow}
    Let $u$ be a bounded nonnegative solution of
    \begin{equation}\label{Heat equation}
        \begin{aligned}
        \begin{cases}
        \partial_t u=\Delta u, &\text{in $ (0,\infty)\times\R^N$}\\
        u(0,x)=u_0(x), &\text{in $\R^N$}
    \end{cases}
        \end{aligned}
    \end{equation}
where $u_0$ is a bounded nonnegative function in $\R^N$. Then $u(t,\cdot)$ is log-concave in $\R^N$ for any $t>0$ if $u_0$ is log-concave.
\end{lemma}
\begin{proof}
   This lemma directly follows as a corollary of the convolution preservation property of log-concavity, since the solution \( u(t,x) \) to the equation \eqref{Heat equation} can be expressed as the convolution of the heat kernel with \( u_0 \) and the heat kernel is log-concave.
\end{proof}
\begin{lemma}\label{log-concavity of a harmonic equation}
(Theorem 4.1 in \cite{kennington1985power})Let $\eta$ solve the harmonic equation:
\begin{equation}
    \begin{aligned}
        \begin{cases}
            -\Delta\eta=1, &\text{in $\Omega$}\\
            \eta>0, &\text{in $\Omega$}\\
            \eta=0, &\text{on $\partial\Omega$}
        \end{cases}
    \end{aligned}
\end{equation}
Then $\eta$ is $1/2$-concave in $\Omega$, which implies that $\log\eta$ is concave in $\Omega$, i.e. $\eta$ is log-concave in $\Omega$. 
\end{lemma}

Here we say a nonnegative function $u$ in $\R^N$ is $p$-concave if 
$$u((1-s)x+sy)\geqslant M_{p}(u(x),u(y);s)$$
for $s\in [0,1]$ and $x,y\in\R^N$, where
\begin{equation*}
    \begin{aligned}
        M_{p}(a,b;\lambda)=
        \begin{cases}
            \max\{a,b\}, &p=+\infty\\
            [(1-\lambda)a^p+\lambda b^p]^{1/p}, &\text{for }p\neq -\infty,0,+\infty\\
            a^{1-\lambda}b^{\lambda}, &p=0\\
            \min\{a,b\}, &p=-\infty
        \end{cases}
    \end{aligned}
\end{equation*}
is the ($\lambda$-weighted) $p$-mean of $a$ and $b$. A simple consequence of Jensen’s inequality is that   
$$M_{p}(a,b;\lambda)\leqslant M_{q}(a,b;\lambda) \text{ if }p\leqslant q.$$
Hence $q$-concave can deduce $p$-concave for any $q\geqslant p$, and we can notice that $0$-concave is equivalent to log-concave. 

\begin{lemma}\label{find a log-concave convergent sequence to u_0}
Let $\Omega$ be a bounded smooth convex domain, $u_0$ be a bounded
nonnegative log-concave function in $\Omega$. Then there exists a sequence of log-concave functions $\{u_{0,n}\}_{n\geqslant 1}$ converging to $u_0$ almost everywhere, such that $u_{0,n}$ is nonnegative and continuous on $\Bar{\Omega}$ with zero boundary value, i.e. $\left.u_{0,n}\right|_{\partial\Omega}=0$, for any $n$.
    
\end{lemma}

\begin{proof}
   Consider the heat function:
   \begin{equation*}
       \begin{aligned}
           \begin{cases}
        \partial_t v=\Delta v, &\text{in $ (0,\infty)\times\R^N$}\\
        v(0,x)=v_0(x), &\text{in $\R^N$}
    \end{cases}
       \end{aligned}
   \end{equation*}
where
\begin{equation*}
    \begin{aligned}
        v_0(x)=\begin{cases}
            u_0(x), &\text{for $x\in\Omega$}\\
            0, &\text{for $x\notin\Omega$}
        \end{cases}.
    \end{aligned}
\end{equation*}
Then the solution $v(t,x)=\left[e^{t\Delta}v_0\right](x)$ for $x\in\R^N$ and $t>0$. Here $e^{t\Delta}$ is the heat semigroup and we let $e^{-\infty}:=0$. Then $v_0$ is log-concave in $\R^N$. We deduce from Lemma \ref{log-concavity of heat flow} that $v(t,\cdot)$ is log-concave for any $t>0$. Furthermore, by the maximum principle, we can deduce from $v_0\in L^1(\R^N)\cap L^{\infty}(\R^N)$ that 
\begin{align}
    &\text{$v(t,\cdot)$ is a positive continuous function in $\R^N$ for any $t>0$,} \label{v(t,x) is positive continuous}\\
        &\text{$\norm{v(t,\cdot)}_{L^{\infty}(\R^N)}<$
     $\norm{v_0}_{L^{\infty}(\R^N)}$ for any $t>0$,}\label{maximum principle}\\
     &\text{$\lim\limits_{t\to 0^+}\norm{v(t,\cdot)-v_0}_{L^1(\R^N)}=0$}.
     \label{L1 convergence in heat equations}
\end{align}
By \eqref{L1 convergence in heat equations}, we can find a sequence $\{t_n\}\subset (0,\infty)$ with $\lim\limits_{n\to\infty}t_n=0$ such that
\begin{equation}\label{a.e.convegence in solution to heat equations}
\lim\limits_{n\to\infty} v(t_n,x)=v_0(x)    
\end{equation}
for almost all $x\in\R^N$.

Let $\eta$ solve this harmonic equation:
\begin{equation*}
    \begin{aligned}
         \begin{cases}
            -\Delta\eta=1, &\text{in $\Omega$}\\
            \eta>0, &\text{in $\Omega$}\\
            \eta=0, &\text{on $\partial\Omega$}
        \end{cases}.
    \end{aligned}
\end{equation*}
Then by Lemma \ref{log-concavity of a harmonic equation}, we obtain that $\eta$ is log-concave and $\log\eta\to-\infty$ as $\text{dist}(x,\partial\Omega)\to 0$. By \eqref{v(t,x) is positive continuous} and \eqref{maximum principle} we can find a sequence $\{m_n\}\subset (1,+\infty)$ with $\lim\limits_{n\to\infty}m_n=+\infty$ such that 
$$V_n(x):=\log v(t_n,x)+m_n^{-1}\log \eta(x)$$
is continuous and concave in $\Omega$ and 
$$\sup\limits_{x\in\Omega}V_n(x)\leqslant 
\mathop {\mathrm{ess}\; \mathrm{sup}} \limits_{x\in\Omega}\log v_0$$
as long as we take
$$m_n \geqslant \left|\dfrac{\log\eta(x)}{\mathop {\mathrm{ess}\; \mathrm{sup}} \limits_{x\in\Omega}\{\log v_0(x)-
\log v(t_n,x)\}}\right|+1.$$
Furthermore, by \eqref{a.e.convegence in solution to heat equations} we have 
\begin{align}
    & \lim\limits_{n\to\infty} V_n(x)=\log v_0(x)= \log u_0(x) \text{ for almost all $x\in\Omega$,}\\
    &V_n(x)\to -\infty\text{ as $\text{dist}(x,\partial\Omega)\to 0$.}
\end{align}
Then the function $u_{0,n}(x):=e^{V_n(x)}$ satisfies
\begin{align}
    & \lim\limits_{n\to\infty} u_{0,n}(x)=u_0(x) \text{ for almost all $x\in\Omega$,}\\
    & u_{0,n}(x)\text{ is log-concave in }\Omega \text{ and continuous on }\Bar{\Omega}
\text{ for }\forall n\\
    & u_{0,n}\geqslant 0 \text{ in }\Omega
    \text{ and } u_{0,n}=0\text{ on }\partial\Omega.
\end{align}
\end{proof}

Based on lemmas above, we then prove Lemma \ref{step2}. Consider the Schr\"odinger equation:
 \begin{equation*}
        \begin{cases}
            \partial_t \gamma + c(x)\gamma-\Delta \gamma = 0 & in \;\; (0,\infty)\times\Omega\\
            \gamma(0,x)=\gamma_0(x) & in \;\; \overline{\Omega}
        \end{cases}
    \end{equation*}
and we do not assume that $\gamma_0=0$ on $\partial\Omega$. By Lemma \ref{find a log-concave convergent sequence to u_0}, we can find a sequence of log-concave functions $\{\gamma_{0,n}\}$ with zero boundary value converge to $\gamma_0$ for almost all $x\in\Omega$. Let
$$\gamma_n(t,x):=\int_{\Omega}K_{\Omega}(x,y,t)\gamma_{0,n}(y)dy$$
where $K_{\Omega}(x,y,t)$ is the kernel function of this Schr\"odinger equation according to \cite{avramidi2015heat}. Then we apply the Lebesgue dominated convergence theorem to obtain 
\begin{equation*}
    \begin{aligned}
        \lim\limits_{n\to\infty} \gamma_n(t,x)=\int_{\Omega} K_{\Omega}(x,y,t)\gamma_0(y)dy=\gamma(t,x),\;\;x\in\Omega,t>0.
    \end{aligned}
\end{equation*}
On the other hand, by Lemma \ref{step1} we see that $\gamma_n(t,\cdot)$ is log-concave in $\Omega$ for any $t>0$. Then we prove Lemma \ref{step2} by the preservation property of log-concavity under convergence in distribution.

\section{Proof of Theorem \ref{step3}}
\label{appendix: proof of step3}
For any smooth convex domain $\Omega$, there exist a sequence of bounded smooth convex domains $\{\Omega_n\}_{n\geqslant 1}$ such that
$$\Omega_1\subset \Omega_2\subset\cdots \subset \Omega_n\subset \cdots,\;\; \bigcup_{n=1}^{\infty} \Omega_n = \Omega$$
according to \textit{Theorem 2.7.1} in \cite{schneider2013convex}. If $\Omega=\R^N$, we only need to take $\Omega_n=B_n(0)$, where $B_n(0)$ is a ball with a radius of $n$ centered on the origin.

Consider a sequence of equations on $\Omega_n$:
\begin{equation}\label{eq:Schrodinger eq in Omega_n}
    \begin{aligned}
        \begin{cases}
             \partial_t \gamma_n + c(x)\gamma_n-\Delta \gamma_n = 0 & in \;\; (0,\infty)\times\Omega_n\\
        \gamma_n(t,x)=\gamma(t,x)\raisebox{0.5ex}{$\chi$}_{\overline{\Omega}_n}(x)
            & in \;\; [0,\infty)\times \overline{\Omega}_n
        \end{cases}
    \end{aligned}
\end{equation}
where $\gamma(t,x)$ is the solution to the Schr\"odinger equation in the entire space $\R^N$:
 \begin{equation*}
        \begin{cases}
            \partial_t \gamma + c(x)\gamma-\Delta \gamma = 0 & in \;\; (0,\infty)\times\R^N\\
            \gamma(0,x)=\gamma_0(x) & in \;\; \R^N
        \end{cases}.
    \end{equation*}
Then by the existence and uniqueness of solutions, $\gamma(t,x)\raisebox{0.5ex}{$\chi$}_{\overline{\Omega}_n}(x)$ is the unique solution to \eqref{eq:Schrodinger eq in Omega_n}.

By Lemma \ref{step2}, we know that $\gamma_n(t,\cdot)$ is log-concave in $\Omega_n$ for any $t>0$, and furthermore, $\gamma_n$ is $L^1$ convergent to $\gamma$, which is sufficient to prove that 
\begin{equation}
    \begin{aligned}
        \lim\limits_{n\to\infty}\int_{\R^N\backslash \Omega_n} \gamma(t,x) dx =0
    \end{aligned}
\end{equation}
for any $t>0, x\in \R^N$. By Lemma \ref{stochastic analysis_FP flow}, we have
\begin{equation*}
    \begin{aligned}
         \gamma(t,x)&=\E\left[
\gamma_0(2W_{t/2})e^{-2\int_{0}^{t/2}c(2W_r)dr}
        \right]\left(\frac{x}{2}\right)\\
        &\leqslant \E\left[
        \gamma_0(2W_{t/2})
        \right]\left(\frac{x}{2}\right)\\
        &=
        \int_{\R^N} \gamma_0(2y) (\pi t)^{-\frac{N}{2}}e^{-\frac{\norm{y-x/2}^2}{t}} dy\\
        &= \int_{\R^N} \gamma_0(y) (4\pi t)^{-\frac{N}{2}}e^{-\frac{\norm{y-x}^2}{4t}} dy.
    \end{aligned}
\end{equation*}
Since $\gamma_0$ is a continuous probability measure in $\R^N$, we know that it is bounded and 
$\displaystyle \lim\limits_{n\to\infty}\int_{\R^N\backslash \Omega_n} \gamma_0(x) dx =0$. Then
\begin{equation*}
    \begin{aligned}
        \int_{\R^N\backslash \Omega_n} \gamma(t,x) dx
        &=\pi^{-\frac{N}{2}}\int_{\R^N\backslash \Omega_n} dx
        \int_{\R^N} \gamma_0(x+2\sqrt{t}\eta) e^{-\eta^2} d\eta\\
        &=\pi^{-\frac{N}{2}}
        \int_{\R^N}  e^{-\eta^2} d\eta
        \int_{\R^N\backslash \Omega_n}
        \gamma_0(x+2\sqrt{t}\eta)
        dx\\
\varlimsup\limits_{n\to\infty}\int_{\R^N\backslash \Omega_n} \gamma(t,x) dx 
&\leqslant 
\pi^{-\frac{N}{2}}
        \int_{\R^N}  e^{-\eta^2} d\eta
\varlimsup\limits_{n\to\infty}\int_{\R^N\backslash \Omega_n}
        \gamma_0(x+2\sqrt{t}\eta) dx
        =0.
    \end{aligned}
\end{equation*}
Hence there exists a subsequence $\{\gamma_{n_k}\}$ such that $\gamma_{n_k}\to \gamma, \text{as } n_k\to\infty$ almost everywhere for any $t>0$, since  $L^1$-convergence sequence includes a subsequence which converges almost everywhere. 

 Furthermore, by the comparison principle, we obtain that
\begin{align}
    &\gamma_{n_k}(t,x)\leqslant \gamma_{n_k+1}(t,x)\leqslant \gamma(t,x),\text{ in }\R^N\times (0,\infty),\\
    & \lim\limits_{{n_k}\to\infty}\gamma_{n_k}(t,x)=\gamma(t,x),\text{ in }\R^N\times (0,\infty).
\end{align}
Then we prove that $\gamma(t,\cdot)$ is log-concave in $\Omega$ for any $t>0$. Thus Theorem $\ref{step3}$ is proved.

\end{document}